\documentclass[lettersize,journal]{IEEEtran}
\usepackage{amsmath,amsfonts}
\usepackage{algorithmic}
\usepackage{algorithm}
\usepackage{array}
\usepackage[caption=false,font=normalsize,labelfont=sf,textfont=sf]{subfig}
\usepackage{textcomp}
\usepackage{stfloats}
\usepackage{url}
\usepackage{cite}
\usepackage{verbatim}
\usepackage{graphicx}
\usepackage{amsthm}
\usepackage{booktabs}
\usepackage{xcolor}
\begin{document}

\title{Distributed Neurodynamics-Based Backstepping Optimal Control for Robust Constrained Consensus of Underactuated Underwater Vehicles Fleet}

\author{Tao Yan,~\IEEEmembership{Graduate Student Member,~IEEE}, Zhe Xu,~\IEEEmembership{Member,~IEEE}, Simon X. Yang,~\IEEEmembership{Senior Member,~IEEE}, \\ S. Andrew Gadsden,~\IEEEmembership{Senior Member,~IEEE}
\thanks{This work was supported by the Natural Sciences and Engineering Research Council (NSERC) of Canada. \textit{(Corresponding author: Simon X. Yang.)}}
\thanks{T. Yan and S. X. Yang are with Advanced Robotics and Intelligent Systems (ARIS)
Laboratory, School of Engineering, University of Guelph, Guelph, ON
N1G2W1, Canada (e-mail: tyan03@uoguelph.ca; zxu02@uoguelph.ca; syang@uoguelph.ca). }
\thanks{Z. Xu and S. A. Gadsden are with Intelligent and Cognitive Engineering (ICE) Laboratory, Department of Mechanical Engineering, McMaster University, Hamilton, ON L8S4L8, Canada (e-mail: xu804@mcmaster.ca; gadsden@mcmaster.ca).}}



\maketitle

\begin{abstract}
Robust constrained formation tracking control of underactuated underwater vehicles (UUVs) fleet in three-dimensional space is a challenging but practical problem. To address this problem, this paper develops a novel consensus based optimal coordination protocol and a robust controller, which adopts a hierarchical architecture. On the top layer, the spherical coordinate transform is introduced to tackle the nonholonomic constraint, and then a distributed optimal motion coordination strategy is developed. As a result, the optimal formation tracking of UUVs fleet can be achieved, and the constraints are fulfilled. To realize the generated optimal commands better and, meanwhile, deal with the underactuation, at the lower-level control loop a neurodynamics based robust backstepping controller is designed, and in particular, the issue of "explosion of terms" appearing in conventional backstepping based controllers is avoided and control activities are improved. The stability of the overall UUVs formation system is established to ensure that all the states of the UUVs are uniformly ultimately bounded in the presence of unknown disturbances. Finally, extensive simulation comparisons are made to illustrate the superiority and effectiveness of the derived optimal formation tracking protocol.
\end{abstract}


\section{Introduction}
\IEEEPARstart{A}{utonomous} underwater vehicle (AUV) is a sort of marine mechatronics systems, and has been used to perform various underwater missions without human intervention \cite{8809889,4,10036108}. Recently, employing a group of autonomous underwater vehicles has attracted growing attention as multi-agent systems are proved to be more efficient, flexible and cost-effective compared to a single AUV, and appear also to be more robust when faced with disturbances or even faults. The main technical problems of this type of system lie in designing effective and efficient coordination protocols to make teams of AUVs perform tasks together in complicated marine conditions. In particular, formation control has increasingly become a focus in multiple AUVs coordination, considering its wide applications in practice. However, formation control of a group of AUVs is barely an easy thing to do, due to the nonlinear, uncertain and underactuated characteristics of dynamics, communication constraints as well as detrimental marine environments. Therefore, it is still an open and pressing problem for both societies of control and ocean engineering \cite{6,9709103,YANreview}.

Formation control of an AUVs fleet can be roughly separated into two major portions, that is, motion coordination and control. For the former, there are several structures and strategies commonly used to coordinate the motion of multiple AUVs, such as leader-following structure \cite{9694518,9557752}, virtual structure\cite{9,8294296}, and artificial potential field approach\cite{13}, etc. Besides coordination, efficient formation controllers are also key to achieving coordinated motions successfully. Following this procedure, extensive research efforts have been made in recent decades in order to synthesize effective and practical formation control protocols for AUVs fleet. In \cite{6547175}, a H$_2$/H$_\infty$ control scheme was proposed based on the leader-following structure to ensure the optimal formation performance when disturbances and communication delays may happen. While the linear quadratic based optimal control is fairly efficient at specific operating points, it may become restrictive when a wide range of operations is required, such as following a time-varying dynamic trajectory. In this respect, nonlinear control techniques have played an important role in the design of high-performance AUVs formation controllers and have been widely applied \cite{gao2018adaptive,shen2017trajectory,li2016receding}. To handle the nonlinearity and underactuation, an adaptive backstepping controller was synthesized with neural network approximation to drive a group of underactuated underwater vehicles (UUVs) to create formation via a leader-following structure \cite{gao2018adaptive}. By incorporating a data-driven predictor, the resulting formation control strategy also addressed the communication delay \cite{9844297}. System constraints fulfillment is another critical concern in designing practical controllers. As such, the model predictive control (MPC) method, as one of the optimal control techniques, was applied to resolve the AUV trajectory tracking problem subject to constraints \cite{shen2017trajectory}. In their studies, a Lyapunov-based backstepping nonlinear MPC algorithm was proposed with stability and feasibility guarantees. Based on a similar idea, receding horizon formation tracking of multiple UUVs with input limitation were addressed \cite{li2016receding}.

In addition to the nonlinearity and underactuation handling, the marine disturbances (e.g., ocean currents, waves and winds) as well as hydrodynamic effects have significant impacts on the acquirement of robust formation performance. Towards this end, a good many research works take advantage of the sliding mode control (SMC) method due to its great robustness in tackling any matched and bounded disturbances \cite{22, elmokadem2017terminal,GUERRERO2023113375,cui2017extended,cheng2018robust}. In \cite{22}, the authors presented an adaptive sliding mode formation control scheme to address issues of the variable added mass and communication constraints, and the overall closed-loop stability was analyzed using Lyapunov theory. To pursue a fast transient performance, a terminal SMC method was adopted for the tracking control of UUVs \cite{elmokadem2017terminal}. While SMC-based approaches are expected to obtain good robustness against the disturbances, the chattering phenomenon stops their applications from real AUV control implementation. To overcome this drawback, a higher order SMC method was proposed for chattering-free trajectory tracking control of AUVs \cite{GUERRERO2023113375}. With the integration of a neurodynamics model, a distributed bioinspired SMC scheme was proposed to address the robust formation tracking of a fleet of fully actuated AUVs \cite{YANbfc}. Other than sliding mode control strategies, observer techniques are other effective alternatives to the improvement of system robustness \cite{GUob}. In \cite{cui2017extended}, considering both unknown disturbances and uncertain nonlinearity, an extended state observer based integral SMC scheme was proposed for an underwater robot, and real-world experiments verified its effectiveness. Employing a similar technique, active disturbance rejection control was used in the dynamic controller design of multiple AUVs formation \cite{9486951}. In \cite{cheng2018robust}, the authors addressed robust finite-time consensus formation control of nonholonomic wheeled mobile robots, in which a finite-time observer was designed to estimate both velocities and disturbances, followed by an integral SMC controller. Likewise, based on a disturbance observer, distributed formation tracking for a group of underactuated AUVs in the horizontal plane was studied \cite{wang2021distributed}. 

While a vital amount of research results as mentioned above have been attained to study the formation control of AUVs fleet, there still are several aspects not well considered from a practical control point of view. Most of the existing AUVs formation protocol adopts a leader-following structure  \cite{6547175,gao2018adaptive,li2016receding,9486951,wang2021distributed}. In such an approach, it is assumed that all the vehicles can have access to the leaders' information, which would be rather restricted in reality. Besides, the robustness analysis in their methods is usually neglected, but it is quite crucial to maintain the feasibility of a method when faced with uncertainties. To handle the underactuation, many existing works follow a backstepping control design procedure, whereas such a method necessitates the derivative of designed virtual commands which is hard to obtain and its robustness is also limited. In terms of disturbance rejection, while sliding mode control behaves well for certain bounded disturbances \cite{22, elmokadem2017terminal,GUERRERO2023113375}, such a method essentially employs a high-gain strategy, thus intrinsically sensitive to the noise. The observer technique acts as an active disturbance compensation \cite{GUob}, yet its performance relies closely on the accurate modeling of particular disturbances, which is almost impossible for the marine situation. On the other hand, system constraints handling and performance optimization are also significant dimensions in the control design of real mechatronic systems, but barely resolved in the existing formation control literature.

Motivated by the above observations, this paper is concerned with the UUVs optimal formation tracking control with unknown disturbances as well as system constraints in three-dimensional (3D) space. Such a problem, clearly, is of practical interest but more challenging, and has not been well studied yet. The main contributions and novelties of this paper are detailed below:
\begin{enumerate}
\item{ A distributed robust optimal protocol is developed for the consensus formation tracking of a fleet of underwater autonomous vehicles in 3D space. The controlled plant is subject to velocity constraints, underactuation, and unknown disturbances.  }
\item { To deal with the underactuation, a spherical coordinate transformation is used, followed by a consensus based formation tracking design. Furthermore, to achieve optimal coordination and meanwhile fulfill the constraints, an on-line motion optimization procedure is developed, and the stability, feasibility, and real-time applicability are discussed.}
\item{To realize the planned optimal commands efficiently and robustly, a neurodynamics based backstepping controller is designed, in which the issue of “explosion of terms” is avoided and control performance is improved. Moreover, the stability and robustness properties are analyzed.}
\item{The overall stability result of the proposed UUVs formation system is derived, which shows that under some moderate conditions, all the states of the UUVs in the fleet can be steered into an ultimate bound even when faced with unknown disturbances. }
\end{enumerate}

The rest of the article is arranged as follows. Some preliminaries are presented in Section \ref{s2}. Section \ref{s3} addresses the constrained consensus formation tracking problem. Neuro-dynamics based robust backstepping controller shall be designed and analyzed in Section \ref{s4}. Section \ref{s5} provides extensive numerical simulations. The conclusion is made in Section \ref{s6}.



\section{PRELIMINARY AND PROBLEM FORMULATION}\label{s2}
In this section, some basic knowledge regarding the graph theory is presented, mathematical models of UUVs are described, and moreover the objective of formation tracking control of UUVs fleet is formulated.

\subsection{Preliminary on graph theory}\label{s2.1}
The communication topology of a UUVs fleet can be modeled by a weighted directed graph $G = \{ {V,E,A} \} $, and each vehicle in such a system can be treated as a node. As for a simple time-invariant graph $G$, it can be described by the vertex set $V = \{ { \nu _1, \nu _2, \ldots, \nu _N } \}$ , edge set $E \subseteq V \times V$, and weighted adjacency matrix $A = \left[ {{a_{ij}}} \right] \in \mathbb{R}^{\text{N} \times \text{N}}$. The element $\nu _i$ in vertex set $V$ denotes $i$-th UUV, and the index $i$ belongs to an index set $\Gamma  = \left\{ {1, \ldots ,N} \right\}$. If  $\nu_i$ is able to receive messages from $\nu_j$ ($i \ne j$), then, say, there exists an edge pointing from $\nu_j$ to $\nu_i$, i.e., $\left( {{\nu _i},{\nu _j}} \right) \in E$, and ${a_{ij}} > 0$; particularly, we call $\nu_j$ a neighbor of $\nu_i$, and all such $\nu_j$ form the set of neighbors of $\nu_i$, denoted by ${N_i} = \left\{ {j | {\left( {{\nu _i},{\nu _j}} \right) \in E} } \right\}$. Otherwise, there is no edge from $\nu_j$ to $\nu_i$, and $a_{ij} = 0$. Moreover, we define $a_{ii} = 0$ for all $i \in \Gamma$, and out-degree $d_{i} = \sum _{j \in N_i} {a_{ij}}$ associated with the node $i$. Then, the degree matrix and Laplacian matrix of graph $G$ are defined as $D = \text{diag} \left\{ {d_1, \ldots, d_N}  \right\} \in \mathbb{R} ^{\text{N} \times \text{N}}$ and $L = D-A$, respectively. A path in $G$ is defined by a set of successive adjacent nodes, starting from any $\nu_i$ and ending at $\nu_j$. If there is at least one path on any two nodes in graph $G$, then, say, graph $G$ is connected.

In order to make the UUVs fleet move along with a prescribed trajectory together, a reference must be defined ahead of time. The availability to the information of reference trajectory for $i$-th UUV is indicated by a parameter $b_i$; that is, if UUV $i$ have access to this information, then $b_i >0$; otherwise, $b_i = 0$, and define matrix  $B = \text{diag} \left( {b_1, \ldots, b_N} \right)$.
\newtheorem{assumption}{Assumption}
\begin{assumption} \label{assumption1}
For the considered multi-UUV formation control network, graph $G$ is connected, and moreover there is at least one UUV able to receive the information of reference trajectory, i.e., the elements of matrix $B$ are not all equal to zero.
\end{assumption}

\newtheorem{lemma}{Lemma}
\begin{lemma}\label{lemma1}
if Assumption \ref{assumption1} holds, then matrix $L+B$ is positive definite.
\end{lemma}

\begin{figure}[!t]
\centering
\includegraphics[width=3in]{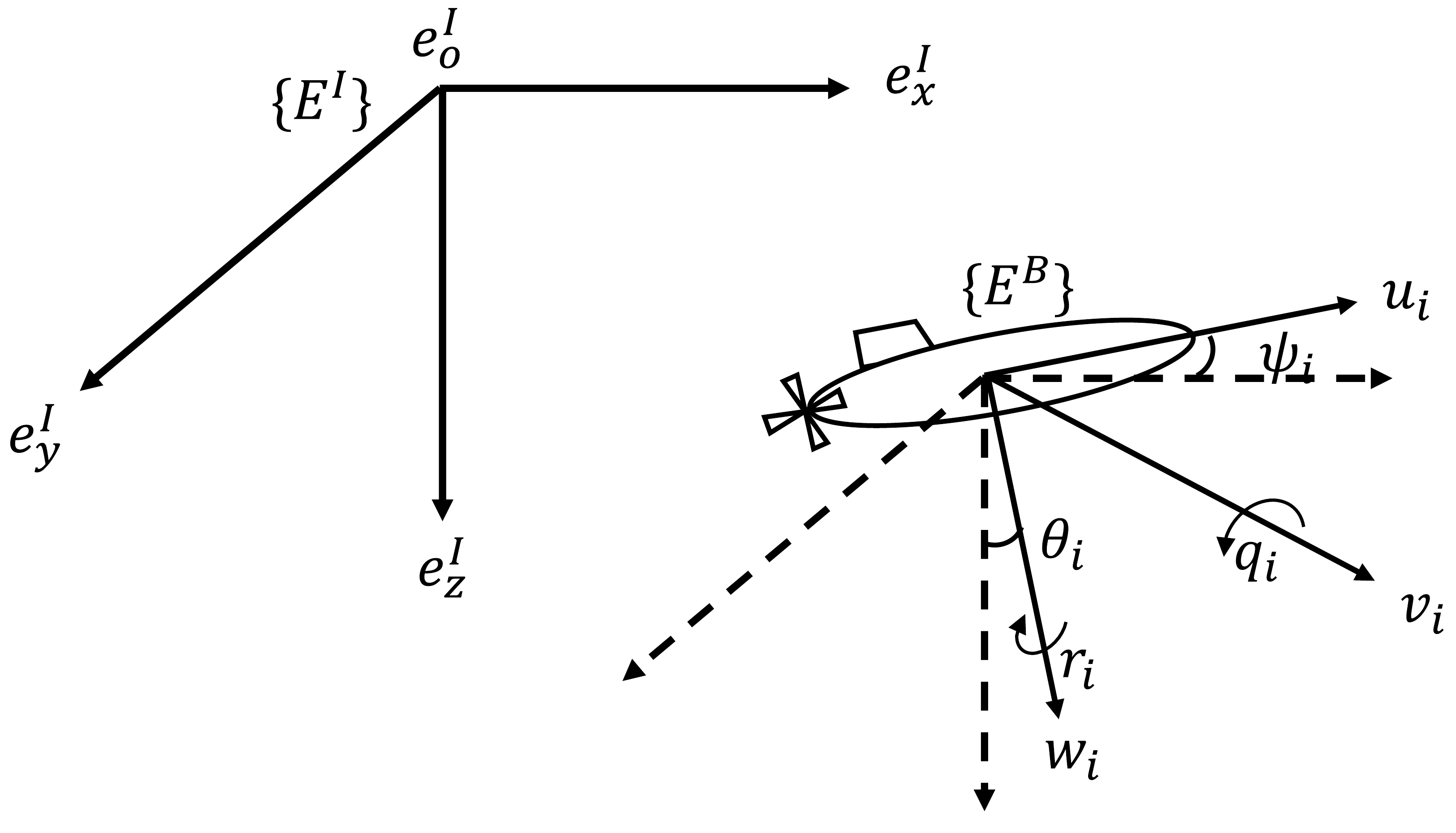}
\caption{Schematic diagram of $i$-th UUV .}
\label{fig1}
\end{figure}

\subsection{Dynamic model and problem formulation}
The distributed robust constrained formation tracking control of fleets of underactuated autonomous underwater vehicles in 3D space is addressed in this paper. First, following the work of Qi \textit{et al.} \cite{qi}, the kinematics of each underwater vehicle are described as
\begin{align}
    \dot x_i &= \cos{\theta _i} \cos{\psi _i} u_i - \sin{\psi _i} v_i + \sin{\theta _i} \cos{\psi _i} w_i, \nonumber \\
    \dot y_i &= \cos{\theta _i} \sin{\psi _i} u_i + \cos{\psi _i} v_i + \sin{\theta _i} \sin{\psi _i} w_i, \nonumber \\
    \dot z_i &= -\sin{\theta _i} u_i + \cos{\theta _i} w_i, \nonumber \\
    \dot \theta _i &= q_i ,\nonumber \\
    \dot \psi _i &= \frac{1}{\cos{\theta _i}} r_i , \label{eq1}
\end{align}
where $\eta _{i1} = [x_i, y_i, z_i]^{\rm T} \in \mathbb{R}^{3}$ and $\eta _{i2} = [\theta _i, \psi _i ]^{\rm T} \in \mathbb{R}^{2}$ represent the location and orientation of the $i$-th vehicle ($i \in \Gamma$), respectively, expressed in the earth-fixed frame $ E^{\rm I} = \left\{ { e_o^{\rm I}, e_x^{\rm I}, e_y^{\rm I}, e_z^{\rm I}} \right\}$, and $\nu _{i1} = [u_i, v_i, w_i]^{\rm T} \in \mathbb{R}^{3}$ and $\nu _{i2} = [q_i, r_i]^{\rm T} \in \mathbb{R}^{2}$ are the linear and angular velocities, respectively, which is expressed in the body-fixed frame $ E^{\rm B} = \left\{ { e_{o,i}^{\rm B}, e_{x,i}^{\rm B}, e_{y,i}^{\rm B}, e_{z,i}^{\rm B}} \right\}$, as shown in Fig. \ref{fig1}.

The dynamics of the $i$-th vehicle is modeled by
\begin{align}
    m_{i1} \dot{u}_i &= m_{i2} v_i r_i - m_{i3} w_i q_i - \beta _{ui} u_i + \tau _{i1} + d_{i1}, \nonumber \\
    m_{i2} \dot{v}_i &= - m_{i1} u_i r_i - \beta _{vi} v_i + d_{i2}, \nonumber \\
    m_{i3} \dot{w}_i &= m_{i1} u_i q_i - \beta _{wi} w_i + d_{i3}, \nonumber \\
    m_{i4} \dot{q}_i &= \left( { m_{i3} - m_{i1} }\right) u_i w_i - \beta _{qi} q_i - \beta _{bi} \sin{\theta _i}+ \tau _{i2} + d_{i4}, \nonumber \\
    m_{i5} \dot{r}_i &= \left( { m_{i1} - m_{i2} }\right) u_i v_i - \beta _{ri} r_i + \tau _{i3} + d_{i5}, \label{eq2}
\end{align}
where $m_{i1} = m_i - \beta _{\dot u i}$, $m_{i2} = m_i - \beta _{\dot v i}$, $m_{i3} = m_i - \beta _{\dot w i}$, $m_{i4} = I_{yi} - \beta _{\dot q i}$ and $m_{i5} = I_{zi} - \beta _{\dot r i}$; $m_i$ is the mass of the $i$-th vehicle; $I_{yi}$ and $I_{zi}$ are the moments of inertia around the axes of $ e_{y,i}^{\rm B}$ and $ e_{z,i}^{\rm B}$, respectively; $\beta _{\left(\cdot \right)}$ is a set of hydrodynamics related terms associated with the $i$-th vehicle. $\tau _i = \left[{\tau _{i1},\tau _{i2},\tau _{i3} }\right]^{\rm T} \in \mathbb{R}^{3}$ is the control input, and $d_i = \left[ {d_{i1},d_{i2},d_{i3},d_{i4},d_{i5} } \right]^{\rm T} \in \mathbb{R}^{5}$ is the unknown disturbance acting on the $i$-th vehicle.

\newtheorem{remark}{Remark}

\begin{remark} \label{rmk2}
It can be clearly seen from \eqref{eq2} that the velocities in sway and heave directions are underactuated; that is, these two degrees of freedom cannot be manipulated directly, and thus the control of such a system can be more challenging.
\end{remark}

To deal with the underactuated constraint, a spherical coordinate transformation \cite{li20203d} is introduced as follows for $i$-th UUV
\begin{align} \label{eq3__}
    u_{ia} &= \sqrt{u_i^2 + v_i^2 + w_i^2}, \nonumber \\
    \theta _{ia} &= \theta _i + \theta _i^\prime , \nonumber \\
    \psi _{ia} &= \psi _i + \psi _i^\prime ,
\end{align}
with
\begin{align}\label{eq4__}
\theta _i^\prime &= \arctan{\left({-w_i}/{\sqrt{u_i^2+v_i^2}}\right)}, \nonumber \\
\psi _i^\prime &= \arctan\left(v_i/u_i\right).
\end{align}
Since $u_{ia}$ is positive, it is easy to verify that $\theta _i^\prime$ and $\psi _i ^\prime$ are both well defined in the open interval $(-\pi/2, \pi/2)$. Applying the above transformation, the kinematics of UUV $i$ in \eqref{eq1} becomes
\begin{align}
    \dot  x_i &= u_{ia} \cos{\theta _{ia}} \cos{\psi _{ia}}, \nonumber \\
    \dot  y_i &= u_{ia} \cos{\theta _{ia}} \sin{\psi _{ia}}, \nonumber \\
    \dot  z_i &= - u_{ia} \sin{\theta _{ia}}, \nonumber \\
    \dot \theta _{ia} &= q_i + \dot \theta _i^\prime, \nonumber \\
    \dot \psi _{ia} &= r_i / \cos{\theta _i} + \dot \psi _i^\prime, \label{eq5__}
\end{align}
Clearly, the transformed variables $\left(u_{ia}, \theta _{ia}, \psi _{ia} \right)$ now are all fully actuated.

In the considered formation tracking problem, the desired geometric formation shape of the UUVs fleet can be determined by a set of relative deviations between the vehicles $i$ and $j$ ($i, j \in \Gamma$), and denote as $\Delta _{ij} = \left[ {\delta_{x,ij}, \delta_{y,ij}, \delta_{z,ij}  }\right]^{\rm T} \in \mathbb{R}^3$ where $\delta_{(\cdot),ij}$ is the relative deviation in a particular direction. In addition to the formation keeping, a prescribed reference trajectory requires to be tracked by the UUVs. Denote by $\eta _{i1}^d = \left[ {x_i^d, y_i^d,z_i^d}\right]^{\rm T} \in \mathbb{R}^3$ the desired 3D trajectory for each vehicle to track. Both variables $\Delta _{ij}$ and $\eta _{i1}^d$ will be given for a particular formation tracking task. We may have the following assumptions.
\begin{assumption} \label{asp2}
 The reference signals, i.e., $\eta _{i1}^d = \left[ {x_i^d, y_i^d,z_i^d}\right]^{\rm T}$ and its first and second derivatives $\dot \eta _{i1} ^ d =  \left[ {\dot x_i^d, \dot y_i^d, \dot z_i^d}\right]^{\rm T}$ and $\ddot \eta _{i1} ^ d =  \left[ {\ddot x_i^d, \ddot y_i^d, \ddot z_i^d}\right]^{\rm T}$, are all bounded for all time ($i \in \Gamma$).
\end{assumption}

\begin{assumption} \label{asp3}
There is some positive constant $\alpha _1$ such that the unknown disturbance $d_i$ enforced on $i$-th UUV is bounded by $ \left\| {d_i} \right\| \le  \alpha _1$.
\end{assumption}

Consider a multiple UUVs formation system where the motion of each individual vehicle is described by equations \eqref{eq1} and \eqref{eq2}, satisfying Assumptions \ref{asp2} and \ref{asp3}. The control objective of this paper is to provide a distributed robust constrained solution for each vehicle such that the following coordination motion can be achieved: 1) The desired formation shape (i.e., desired deviations $\Delta _{ij}$) can be formed and maintained by UUVs. 2) Besides, the UUV fleet can track a predefined trajectory together even in the presence of disturbances. 3) The restrictions in velocities and control inputs should be realized.

\section{Constrained Formation Tracking Control Protocol Design} \label{s3}
To achieve the preceding control requirements, this section addresses the consensus based formation tracking control problem for a fleet of underactuated underwater vehicles, in which all neighbors' information is considered. In particular, the control commands of each vehicle are optimized through an on-line motion optimization procedure so that the planned maneuvering actions could be ensured within a practical range, meanwhile realizing the required specifications.

\subsection{Distributed formation tracking controller}
Let us first define the consensus formation tracking error for $i$-th UUV $(i\in \Gamma)$ as follows
\begin{align} \label{eq4}
      e_{i} = \sum\limits_{j \in N_i} {a_{ij} \left(\eta _{i1}- \eta _{j1} -\Delta _{ij} \right) + b_i \left(\eta _{i1} -\eta _{i1} ^d \right)}, 
\end{align}
where the non-negative indicator $a_{ij}$ shows the information interactions between vehicle $i$ and its neighbors $j \in N_i$, and non-negative constant $b_i$ indicates whether or not the $i$-th vehicle can access the information of the reference trajectory. $\Delta _{ij}$ denotes the desired constant relative position between vehicles $i$ and $j$. Taking the time derivative of equation \eqref{eq4}, yield
\begin{align} \label{eq5}
    \dot e_{i} = \sum\limits_{j \in N_i} {a_{ij} \left(\dot \eta _{i1}- \dot \eta _{j1} \right) + b_i \left(\dot \eta _{i1} - \dot \eta _{i1} ^d \right)}.
\end{align}

Based on the above consensus error defined for each individual, we denote with $e = \left[ e_1^{\rm T},e_2^{\rm T}, \ldots, e_N^{\rm T} \right]^{\rm T} $the consensus fromation tracking error of overall UUVs formation system and its time derivative $ \dot e = \left[ \dot e_1^{\rm T},\dot e_2^{\rm T}, \ldots, \dot e_N^{\rm T} \right]^{\rm T} $. Since UUVs fleet moves as a whole, the desired velocities are the same (i.e., $\dot{\eta} _{i1} ^d = \dot{\eta}_{1} ^d$ for all $i \in \Gamma$).   Then, a set of equations \eqref{eq5} can be arranged into a compact form
\begin{equation}\label{eq6}
    \dot{e} = \left( { L+B }\right) \left( { \dot{\eta} - \text{1}_{N} \dot{\eta}_1 ^d } \right),
\end{equation}
where $\dot \eta _{1} = \left[ \dot \eta _{11}^{\rm T}, \ldots, \dot \eta _{N1}^{\rm T} \right] ^{\rm T}$, $\text{1}_{N} = \left[ 1, \ldots,  1\right] ^{\rm T}$, and matrices $L$ and $B$ describing the communication topology of the considered formation system are detailed in the previous section (see Section \ref{s2.1}). 

Recall that after transformation the $\eta _{i1}$-dynamics is governed by equation \eqref{eq5__}, and it is desirable to design control commands driving the consensus error $e$ to zero. To this end, define
\begin{align}
\rho _i = \begin{bmatrix} \rho _{ix} \\
                        \rho _{iy}\\
                        \rho_{iz} \end{bmatrix}
    = \begin{bmatrix} u_{ia}^{cmd} \cos{\theta _{ia} ^{cmd}} \cos{\psi _{ia} ^{cmd}}\\
                    u_{ia}^{cmd} \cos{\theta _{ia} ^{cmd}} \sin{\psi _{ia} ^{cmd}}\\
                     - u_{ia}^{cmd} \sin{\theta _{ia} ^{cmd}} \end{bmatrix}, \label{eq10__}
\end{align}
and a virtual control law is proposed for UUV $i$ as 
\begin{align}
    \rho _i = - K_{i1} e_i + \dot{\eta}_1^d, \label{eq11__}
\end{align}
where $K_{i1}\in \mathbb{R} ^ {3 \times 3}$ is a diagonal positive definite matrix. By \eqref{eq10__}, we then get the following control commands for UUV $i$:
\begin{align}
    u_{ia}^{cmd} &= \sqrt{\rho _{ix}^2 +\rho _{iy}^2 +\rho _{iz}^2 }, \nonumber\\
    \theta _{ia}^{cmd} &= -\arcsin{ \left(\frac{\rho _{iz}}  {u_{ia} ^{cmd}} \right)}, \nonumber\\
    \psi _{ia}^{cmd} &= \arcsin{ \left(\frac{\rho _{iy}}  {u_{ia} ^{cmd} \cos{\theta _{ia} ^{cmd}}}\right)}. \label{eq12__}
\end{align}
There may exist multiple results for the above angle commands. Observe that in many practical tasks it is unlikely to specify motions for UUVs such that the $\theta _i$ goes beyond $\pm{\pi/2}$, whereby the angle command $\theta^{cmd}_{ia}$ can select those values within the interval $\left(-\pi /2 , \pi/2\right)$. For the determination of $\psi ^{cmd}_{ia}$, based on the kinematic relation \eqref{eq10__}, we also have the expression $\tan{(\psi ^{cmd}_{ia})} = \rho _{iy} / \rho _{ix} $, by which the unique solution for $\psi ^{cmd}_{ia}$ can be determined.

On the basis of the virtual controller as proposed in \eqref{eq10__}--\eqref{eq12__}, the $\eta _{i1}$-dynamics can be rewritten as
\begin{align}
    \dot{ \eta}_{i1} = - K_{i1} e_i + \dot{\eta}_1^d + \sigma _i, \label{eq13__}
\end{align}
with 
\begin{align}
    \sigma _ i = \begin{bmatrix} u_{ai} \cos{\theta _{ai}} \cos{\psi _{ai}}- u_{ai}^{cmd} \cos{\theta _{ai} ^{cmd}} \cos{\psi _{ai} ^{cmd}}\\
    u_{ai} \cos{\theta _{ai}} \sin{\psi _{ai}} -u_{ai}^{cmd} \cos{\theta _{ai} ^{cmd}} \sin{\psi _{ai} ^{cmd}} \\
    - u_{ai} \sin{\theta _{ai}} + u_{ai}^{cmd} \sin{\theta _{ai} ^{cmd}} \end{bmatrix}. \label{eq14__}
\end{align}
Substituting equation \eqref{eq13__} into \eqref{eq6}, yield
\begin{align}
     \dot{e} = \left( { L+B }\right) \left( { - K_1 e  + \sigma } \right), \label{eq15__}
\end{align}
where $K_1 =\text{diag}(K_{11}, \ldots, K_{N1})$ and $\sigma = [\sigma _1, \ldots, \sigma _N]^{\rm T}$.

We then state the following stability properties.
\begin{lemma} \label{lm2}
The overall error dynamics of multiple UUVs consensus formation tracking as described in \eqref{eq15__}, obtained by applying the protocol \eqref{eq11__} and \eqref{eq12__}, is input-to-state stable with respect to $\sigma$.
\end{lemma}
\begin{proof} \label{pf1}
 Propose the following Lyapunov function candidate 
 \begin{align} \label{eq13}
     V_1 = \frac{1}{2} e^{\rm T}(L+B)^{-1} e.
 \end{align}
It is clear from Lemma \ref{lemma1} that the proposed Lyapunov function is valid. Then, taking the time derivative of $V_1$, along  the error dynamics \eqref{eq15__}, and applying again the Lemma \ref{lemma1} we obtain
 \begin{align}\label{eq16}
     \dot V_1 & = - e^{\rm T} K_1 e + e^{\rm T} \sigma \nonumber \\
            & \le - \underline{k}_1 \| e \|^2 + \| e\| \| \sigma  \| \nonumber \\
            & \le - ( \underline{k}_1 - \varepsilon _1) \| e \| ^2 \quad \text{whenever} \quad \| e \| \ge \mu _1,
 \end{align}
 where $\underline{k}_1$ is the minimum eigenvalue of $K_1$, $\varepsilon _1$ an arbitrary number within the interval $(0, \underline{k}_1)$, and $\mu _1 = (1/ \varepsilon _1) \|\sigma \|$.
 
 It is immediate from \eqref{eq16} that the error dynamics \eqref{eq15__} is input-to-state stable with respect to the input $\sigma$, and in particular if $\| \sigma \| \to 0$ as $t \to  \infty$, then the origin of the error system is asymptotically stable. This completes the proof.\end{proof}
 
\begin{remark} \label{rmk3}
It is noted that the virtual control law $\rho _i$ designed for UUV $i$, as seen in \eqref{eq11__}, uses simply the information from its neighbors, and therefore the resulting formation protocol is said to be fully distributed.
\end{remark}

\subsection{On-Line Motion Optimization Procedure}
It is shown above that the virtual control law proposed can lead to stable consensus formation tracking so long as the error of control commands can be made bounded. It is, however, inevitable to tune the virtual control gain $K_i$ carefully so as to meet a satisfactory consensus formation tracking performance, and besides once the control gain is determined it cannot be changed in all future time. These features may greatly restrict the performance of the present formation plan in reality. To relax it, this subsection develops an on-line optimization procedure such that the control gain $K_i$ can be optimized automatically with respect to a certain performance index; meanwhile, the constraints on UUVs' velocities can be fulfilled to improve the efficacy and security of the resulting optimal control actions.

The optimal virtual control gain of $i$-th UUV can be obtained by solving the following constrained minimization problem at sampling time instant $t_k $, $(t_k > 0)$:
\begin{align}
    \min _{K_{i1,1}} \ J_i &= e_{i,1} ^{\rm T} Q e_{i,1} + \rho _{i,1}^{\rm T} R_1 \rho _{i,1} +  p_1 ({k}_{i1,1}-{k}_{i1,0})^{2}  \nonumber \\
    &\quad + p_2 ({k}_{i2,1}-{k}_{i2,0})^{2}+ p_3 ({k}_{i3,1}-{k}_{i3,0})^{2}  \nonumber\\
    &\quad + (\rho _{i,1}- \rho _{i,0})^{\rm T} R_2 (\rho _{i,1}-\rho _{i,0}) \label{eq17}\\
     \text{s.t.} \quad & e_{i,1} = \sum\limits_{j \in N_i}({\eta} _{i1,1} - {\eta} _{j1,1}-\Delta_{ij}) \nonumber \\
                    &\qquad \qquad +b_i( \eta _{i1,1} -\eta _{i1,1}^d) \label{eq18}\\
                    & {\eta} _{i1,1} = \eta _{i,0} +  \rho _{i,1} \Delta t \label{eq19}\\
                    & \rho _{i,1} = -K_{i1,1} e_{i,0} + \dot{\eta}_{1,0}^d  \label{eq20}\\
                    &\qquad \underline{\rho}_{i} \le \rho _{i,1} \le \bar \rho_{i}  \label{eq21}\\
                    &\qquad   -K_{i1,1} \prec 0  \label{eq22} \\
                    & \eta _{i1,0} = \eta _{i1} (t_k), \dot{\eta}_{1,0}^d=  \dot{\eta} _{1}^d (t_k), {\eta} _{j1,1} =  \eta _{j1} (t_k) \nonumber  \\
                    & e_{i,0} = e_{i} (t_k), \eta _{i1,1}^d=  \eta _{i1}^d(t_k), \rho _{i,0} = \rho _{i}(t_k) \nonumber\\
                    & {k}_{ij,0} = {k}_{ij}(t_k) , \ (j = 1,2,3), \label{eq23}
\end{align}
where $Q$, $R_1$, $R_2\in \mathbb{R}^{3\times3}$ and $p_1$, $p_2$, $p_3 >0$ are weighting parameters of the suggested quadratic objective function $J_i$; $K_{i1,1} = \text{diag}(k_{i1,1},k_{i2,1},k_{i3,1})$ is the virtual control gain to be solved; $\Delta t$ is the sampling period; $\underline{\rho}_{i}$, $\bar{\rho}_{i} \in \mathbb{R}^{3} $ are constant vectors used to restrict the speed of the vehicle; $\eta _{i1} (t_k)$, $\dot{\eta} _{1}^d (t_k)$, $\eta _{j1} (t_k)$, $e_{i} (t_k)$, $\eta _{i1}^d(t_k)$, $\rho _{i}(t_k)$ and $ {k}_{ij}(t_k)$ are the initial values of the optimization problem, all of which are sampled at the time instant $t_k $. Notice that in this procedure some additional terms can be tailored in the objective function $J_i$ to achieve extra goals like obstacles and collision avoidance; for example, such a term can be designed as the reciprocal of the distances between UUV and obstacles.

\begin{remark} 
It should be noted that the above constrained minimization problem as in \eqref{eq17}--\eqref{eq23} is implemented in a real-time manner. That is, at each sampling time instant $t_k \ (k > 0)$, based on the measured current states and last optimizing results, the minimization problem is solved independently by each vehicle to yield the optimal solution $K_{i1,1}^\star$. Using this gain, the current optimal control commands can be acquired by applying equations \eqref{eq11__} and \eqref{eq12__}.  At the next sampling time instant $t_{k+1}$, this procedure is repeated, and a new optimal solution will be calculated. This way, the control policy could be optimized dynamically in order to reach the best performance.
\end{remark}

\begin{remark} \label{rmk4}
In designed constrained motion optimization, an approximate predictive model, shown in \eqref{eq19} with sampling time period $\Delta t> 0$, is used so as to be able to generate the one-step ahead predicted trajectories which are being optimized in terms of the consensus error, energy consumption as well as the smoothness of both processes of control and dynamic optimization, as the performance index suggested. As such, the resulting control policy could be more efficient, adaptive and consistent than the previous one \eqref{eq11__} where the control gain $K_{i1}$ is commonly determined by trial-and-error. Besides, the restrictions on UUV's velocities can be fulfilled as well by the constraint \eqref{eq21}, which effectively ensures the practicality and security of the proposed scheme.
\end{remark}

\begin{remark} \label{rmk5}
The feasibility of the constrained optimization problem \eqref{eq17}--\eqref{eq23} is straightforward, as there are no state constraints imposed on UUVs. It is also clear that the stability of the resulting optimized control commands can be guaranteed, which is attributed to the introduced constraint \eqref{eq22}. This follows directly from the result of Lemma \ref{lm2}. Most importantly, unlike the framework\cite{shen2017trajectory}, since the presented problem is fully convex (i.e., both objective function and constraints are convex), there exist highly efficient programming methods (e.g., interior point methods) to solve it without affecting the real-time capability.
\end{remark}

\section{Bioinspired Robust Controller Design} \label{s4}
The constrained consensus formation tracking problem of UUVs fleet is addressed in the previous section. In particular, distributed optimal control commands are derived for each UUV by on-line solving a constrained minimization problem, and the stability and flexibility of the developed optimization problem are clarified. This section investigates the robust dynamic control of underactuated underwater vehicles so that the derived optimal control commands can be realized effectively even in the presence of various unknown marine disturbances.

\subsection{Backstepping Design Procedure}
To address the underactuation issue, this section employs the backstepping design procedure, in which two auxiliary virtual controllers are defined to help design the final control laws. Before the derivations, we define the following error variables
\begin{align}
    \tilde q_i &=  q_i - q_i^{cmd}, \label{eq24} \\ 
    \tilde r_i &=  r_i - r_i^{cmd} , \label{eq25}
\end{align}
where $q_i$ and $r_i$ are the actual body frame angular velocities of UUVs in the yaw and pitch directions, respectively, and $q_i^{cmd}$ and $r_i^{cmd}$ are corresponding two virtual control commands to be designed.

Based on the relations introduced in \eqref{eq24} and \eqref{eq25}, the $(\theta _{ia}, \psi _{ia})$-dynamics in \eqref{eq5__}, become
\begin{align}
    \dot{\theta}_{ia} &= q_i^{cmd} + \tilde{q}_i + \dot{\theta}_i^\prime, \label{eq26}\\
    \dot{\psi}_{ia} &= \left(r_i^{cmd} + \tilde{r}_i\right) / \cos{\theta _i} + \dot{\psi}_i^\prime. \label{eq27}
\end{align}
The virtual controllers now are designed as follows
\begin{align}
    q_i ^{cmd} &= -k_{i\theta} \left({\theta _{ia} - \theta _{ia}^{cmd} }\right) - \dot{\theta _i ^{\prime}} + \dot{\theta} _{ia}^{cmd}, \label{eq28}\\
    r_i^{cmd} &= \cos{\theta _i} \left[ { -k_{i\psi}\left(\psi _{ia} - \psi _{ia}^{cmd}\right) -\dot{\psi}_{i}^{\prime} +\dot{\psi}_{ia}^{cmd}} \right], \label{eq29}
\end{align}
where $k_{i \theta}$ and $k_{i\psi}$ are some positive constants. Let $\tilde \theta _{ia} =\theta _{ia} - \theta _{ia}^{cmd}$, $\dot {\tilde \theta}_{ia} = \dot \theta _{ia} - \dot \theta _{ia}^{cmd}$, $\tilde \psi_{ia} = \psi _{ia} - \psi _{ia}^{cmd}$ and $\dot {\tilde \psi}_{ia} =\dot \psi _{ia} - \dot \psi _{ia}^{cmd}$. Substituting the virtual control laws \eqref{eq28} and \eqref{eq29} into equations \eqref{eq26} and \eqref{eq27}, respectively, yield
\begin{align}
    \dot{\tilde \theta}_{ia} &= - k_{i\theta} \tilde \theta _{ia} + \tilde q_i, \label{eq30} \\
    \dot{\tilde \psi}_{ia} &= - k_{i\psi} \tilde \psi _{ia} +  \tilde r_i / \cos{\theta _i}. \label{eq31}
\end{align}

Taking the time derivatives in \eqref{eq24} and \eqref{eq25}, together with the dynamics of $q_i$ and $r_i$ shown in \eqref{eq2}, we obtain
\begin{align}
    \dot{\tilde{q}}_i &= \left(\frac{1}{m_{i4}}\right) \left[\left( { m_{i3} - m_{i1} }\right) u_i w_i - \beta _{qi} q_i - \beta _{bi} \sin{\theta _i} \right. \nonumber \\
    & \qquad \qquad \qquad\left. + \tau _{i2} + d_{i4} \right] - \dot q_i^{cmd} , \label{eq32}  \\
    \dot{\tilde{r}}_i &=  \left(\frac{1}{m_{i5}}\right) \left[\left( { m_{i1} - m_{i2} }\right) u_i v_i - \beta _{ri} r_i + \tau _{i3} + d_{i5}\right] - \dot r_i^{cmd}. \label{eq33}
\end{align}

Let $\tilde u_{ia} = u_{ia} - u_{ia}^{cmd}$ and $\dot{\tilde u}_{ia} = \dot{u}_{ia} - \dot{u}_{ia}^{cmd}$. Due to the $(u_i, v_i, w_i)$-dynamics in \eqref{eq2} as well as transformation \eqref{eq3__}, we have
\begin{align}
    \dot{\tilde u}_{ia} &=  \frac{\cos{\theta _i^\prime} \cos{\psi _i^\prime}}{m_{i1}} \left( m_{i2} v_i r_i - m_{i3} w_i q_i - \beta _{ui} u_i + \tau _{i1} \right) \nonumber \\
    &\quad - \frac{\cos{\theta _i^\prime} \sin{\psi _i^\prime}}{m_{i2}} \left( m_{i1}u_{i}r_{i} + \beta _{vi} v_i \right) \nonumber \\
    & \quad- \frac{\sin{\theta _i^\prime}}{m_{i3}} \left( {m_{i1} u_i q_i - \beta _{wi}w_i}\right) - \dot{u}_{ia}^{cmd} + \bar{d}_{i1}, \label{eq34}
\end{align}
where
\begin{align*}
    \bar d_{i1} &= \frac{\cos{\theta _i^\prime} \cos{\psi _i^\prime}}{m_{i1}} d_{i1} + \frac{\cos{\theta _i^\prime} \sin{\psi _i^\prime}}{m_{i2}} d_{i2} + \frac{\sin{\theta _i^\prime}}{ m_{i3}} d_{i3}.
\end{align*}
The goal now is to seek control laws for $\tau _{i1}$, $\tau _{i2}$ and $\tau _{i3}$ such that the error variables $\tilde \theta _{ia}$, $\tilde \psi _{ia}$, $\tilde q_i$, $\tilde r_i$ and $\tilde u_{ia}$, governed by equations \eqref{eq30}--\eqref{eq34}, can be brought to the origins. We provide the following lemma to achieve this purpose.

\begin{lemma} \label{lm4}
Consider the error system described by \eqref{eq30}--\eqref{eq34} with Assumption \ref{asp3} being satisfied. The resulting closed-loop error system is input-to-state stable, if the following control laws are employed
\begin{align}
    \tau _{i1} &=  -m_{i2} v_i r_i + m_{i3} w_i q_i + \beta _{ui} u_i + \frac{ {m_{i1}}} {\cos{\theta _i^\prime} \cos{\psi _i^\prime}} \left( -k_{iu} \tilde u_{ia} \right. \nonumber \\
     &\quad \left. + \dot{u}_{ia}^{cmd} + \tau _{i1}^*\right),  \label{eq35}\\
     \tau _{i1}^*  &= \frac{\cos{\theta _i^\prime} \sin{\psi _i^\prime}}{m_{i2}} \left( m_{i1}u_{i}r_{i} + \beta _{vi} v_i \right) + \frac{\sin{\theta _i^\prime}}{m_{i3}} \left( m_{i1} u_i q_i \right. \nonumber\\
     &\quad \left. - \beta _{wi}w_i\right), \label{eq35_} \\
    \tau _{i2} &=  - \left( { m_{i3} - m_{i1} }\right) u_i w_i + \beta _{qi} q_i + \beta _{bi} \sin{\theta _i} - k_{iq} {m_{i4}} \tilde q_i \nonumber \\
     & \quad  -{m_{i4}}\tilde{\theta}_{ia} + {m_{i4}} \dot{q}_i^{cmd}, \label{eq36} \\
    \tau _{i3} &=  - \left( { m_{i1} - m_{i2} }\right) u_i v_i + \beta _{ri} r_i  - k_{ir} {m_{i5}} \tilde r_i  \nonumber \\
    &\quad -{m_{i5}} \tilde{\psi}_{ia} / \cos{\theta _i} + {m_{i5}} \dot{r}_i^{cmd}, \label{eq37}
\end{align}
with $k_{iu}$, $k_{iq}$, $k_{ir}$ being some positive constants.
\end{lemma}
\begin{proof} \label{pf4}
Plugging the proposed control laws in \eqref{eq35}--\eqref{eq37} into the equations \eqref{eq32}--\eqref{eq34}, respectively, result in the below $(\tilde{q}_i,\tilde{r}_i,\tilde{u}_{ia})$-error dynamic system
\begin{align}
    \dot{\tilde{q}}_i &= -k_{iq} \tilde q_i - \tilde{\theta}_{ia}+ \frac{1}{m_{i4}} d_{i4}, \label{eq38}\\
    \dot{\tilde{r}}_i &= -k_{ir} \tilde r_i - \tilde{\psi}_{ia}/\cos{\theta _i}  + \frac{1}{m_{i5}} d_{i5}, \label{eq39} \\
    \dot{\tilde{u}}_{ia} &= -k_{iu} \tilde u_{ia} +  \bar{d}_{i1} .\label{eq40}
\end{align}
We propose the Lyapunov function candidate as
\begin{align}
    V_{i2} = \frac{1}{2} \tilde{\theta}_{ia}^2 + \frac{1}{2} \tilde{\psi}_{ia}^2 + \frac{1}{2} \tilde{r}_i^2 + \frac{1}{2} \tilde{q}_i^2 +\frac{1}{2} \tilde{u}_{ia}^2, \label{eq41} 
\end{align}
and employing the time derivative of $V_{i2}$ along the trajectories of error system in \eqref{eq30}, \eqref{eq31} and \eqref{eq38}--\eqref{eq40} yields
\begin{align}
    \dot V_{i2} &=  -k_{i\theta} \tilde{\theta}_{ia}^2 -k_{i\psi} \tilde{\psi}_{ia}^2 -k_{ir} \tilde{r}_i^2 -k_{iq} \tilde{q}_i^2 -k_{iu} \tilde{u}_{ia}^2 \nonumber \\
    & \quad  + \frac{1}{m_{i4}} \tilde q_i d_{i4} + \frac{1}{m_{i5}} \tilde r_i d_{i5} + \tilde u_{ia} \bar{d}_{i1}.\label{eq42}
\end{align}
Let  $\zeta _i = [\tilde \theta _{ia}, \tilde \psi _{ia}, \tilde r_i, \tilde q_i, \tilde u_{ia}]^{\rm T}$ and $\tilde{d}_i = \left[0,0, d_{i4}, d_{i5}, \bar{d}_{i1} \right]^{\rm T}$. The equation \eqref{eq42} becomes
\begin{align}
    \dot V_{i2} & \le - \underline{k}_i \|\zeta _i\|^2 + \gamma _i \|\zeta _i\|  \| \tilde{d}_i\|,
\end{align}
with 
\begin{align*}
    \underline{k}_i &= \inf \left\{  k_{i\theta}, k_{i\psi}, k_{ir}, k_{iq}, k_{iu} \right\}, \\
    \gamma _i&= \sup \left\{ \frac{1}{{m_{i4}}}, \frac{1}{{m_{i5}}}, 1\right\}.
\end{align*}
Applying Assumption \ref{asp3}, obtain
\begin{align}
    \dot V_{i2} & \le - \underline{k}_i \|\zeta _i\|^2 + \gamma _i \alpha _1  \|\zeta _i\| \nonumber \\
    & \le - \left( \underline{k}_i - \lambda _1 \right) \|\zeta _i\|^2 \quad  \text{whenever}\quad \|\zeta _i \|\ge \frac{\gamma _i \alpha _1}{\lambda _1}, \label{eq42__}
\end{align}
where $ 0 <\lambda _1 < \underline{k}_i$. This completes the proof.
\end{proof}

\begin{remark} \label{rmk6}
It is observed that in standard backstepping procedure while the stability of the closed-loop system can be guaranteed by the Lemma \ref{lm4}, the derived controllers \eqref{eq35}-\eqref{eq37} rely on the time derivatives of the virtual commands introduced in \eqref{eq28} and \eqref{eq29}. Therefore, the implementation of such control laws may become much complicated with many terms resulting from the differential operation, also known as the issue of "explosion of terms", which poses a great difficulty on applications of backstepping control design. \end{remark}

\begin{remark} \label{rmk7}
In many simulational studies, it is not uncommon to use a numerical difference to approximate the analytical one to simplify the backstepping control realization. However, such a numerical operation can be rather vulnerable to the noise which exists ubiquitously in real processes and, thus, may lead to undesired behavior or even instability of the overall formation system.
\end{remark}

\begin{remark} \label{rmk8}
It is also clear from the properties of $\dot V_{i2}$ as in \eqref{eq42__} that the robustness of the backstepping controller against the disturbances is mainly dependent of the selection of the corresponding control gains (i.e., $k_{i\theta}$, $k_{i\psi}$, $k_{iu}$, $k_{ir}$ and $k_{iq}$). Consequently, due to the cascade connection generated by the backstepping procedure, it is easier to result in a high-gain controller, and then more likely to wind up the actuators in practice.
\end{remark}
Based on the observations, it is necessary to remedy the above backstepping control laws so that more practical and efficient controllers can be synthesized while the nice robustness properties can be obtained without the employment of such high-gain control laws.

\subsection{Neural Dynamics-Based Robust Control Design}
To overcome the aforementioned challenges, a bioinspired solution is provided to improve the robustness properties of the conventional backstepping controllers, and it is noticed that in order to avoid the "explosion of terms", the time derivatives of the auxiliary variables, i.e., $\dot{q}_i^{cmd}$ and $\dot{r}_i^{cmd}$, are regarded as the disturbances in the sequel and counteracted by the introduced neurodynamics model.

Shunting model as one of the biologically inspired neurodynamics models was initially proposed to describe the behavior of neurons in membrane with stimulus. By virtue of its beneficial properties, it has been extensively employed to develop bio-driven autonomous systems. The original shunting model of $i$-th neuron can be described by the following switching nonlinear differential equation
\begin{align} \label{eq43}
    {\dot x_i} =  - {a_i}{x_i} + \left( {{b_i} - {x_i}} \right)s_i^ +  - \left( {{b^\prime_i} + {x_i}} \right)s_i^ - ,
\end{align}
where $x_i$ represents the $i$-th neuron activities, $s_i ^+$ and $s_i ^-$ capture the environmental excitatory and inhibitory inputs, respectively, and $a_i$, ${b}_i$ and $b^\prime_i$ are some positive coefficients. It is noted that when the input $s_i$ to \eqref{eq43} is non-negative, $s_i ^+ = s_i$ and $s_i ^- = 0$; otherwise, $s_i ^+ = 0$ and $s_i ^- = -s_i$ . On this basis, shunting model \eqref{eq43} can be rewritten as 
\begin{align} 
    {\dot x_i} =  - \left( {a_i} + | s_i | \right){x_i} +  {b_i}s_i^ +  - {b^\prime_i} s_i^ - .
\end{align}
Integrated with the above model, the bioinspired robust control laws are then designed as follows
\begin{align}
    \tau _{i1} &=  -m_{i2} v_i r_i + m_{i3} w_i q_i + \beta _{ui} u_i  \nonumber \\
   &\quad + \frac{ {m_{i1}}} {\cos{\theta _i^\prime} \cos{\psi _i^\prime}} \left( -k_{iu} x_{i1} + \tau _{i1}^* \right),  \label{eq43_}\\
   \dot{x}_{i1} &= - \left( a_{i1} + \left| \tilde u_{ia} \right| \right) x_{i1} + g_{i1}\left( \tilde u_{ia} \right),  \label{eq44}\\
    \tau _{i2} &=  - \left( { m_{i3} - m_{i1} }\right) u_i w_i + \beta _{qi} q_i + \beta _{bi} \sin{\theta _i}  \nonumber \\
    & \quad - k_{iq} {m_{i4}} x_{i2} -{m_{i4}}\tilde{\theta}_{ia}, \label{eq45} \\
     \dot{x}_{i2} &= - \left( a_{i2} + \left| \tilde q_i \right| \right) x_{i2} + g_{i2}\left( \tilde q_i\right), \label{eq45_} \\
    \tau _{i3} &=  - \left( { m_{i1} - m_{i2} }\right) u_i v_i + \beta _{ri} r_i - k_{ir} {m_{i5}} x_{i3} \nonumber \\
    &\quad -{m_{i5}} \tilde{\psi}_{ia} / \cos{\theta _i} , \label{eq46}\\
     \dot{x}_{i3} &= - \left( a_{i3} + \left| \tilde r_i \right| \right) x_{i3} + g_{i3}\left( \tilde r_i\right), \label{eq47}
\end{align}
where $a_{i1}$, $a_{i2}$ and $a_{i3}$ are the positive constants to be designed, and functions $g_{ij}, (j=1,2,3)$ are defined as
\begin{align*}
    g_{ij} \left(   y \right) &= \begin{cases}
    b_{ij} y, & y \ge 0\\
    b^\prime_{ij} y, &  y < 0
    \end{cases}
\end{align*}
with $b_{ij}$ and $b^\prime_{ij}$ being some positive constants. Notice that for ease of analysis in the sequel we rewrite the expression of the shunting model in our designed controller, and it is easy to verify that the equations \eqref{eq44}, \eqref{eq45_} and \eqref{eq47} are equivalent to the primal form \eqref{eq43}.

\subsection{Stability Analysis}
Substituting into the equations \eqref{eq32}--\eqref{eq34} the neuro-dynamics based backstepping controllers in \eqref{eq43_}--\eqref{eq47}, we obtain the following modified error subsystems of $\tilde u_{ia}$, $\tilde q_i$ and $\tilde r_i$
\begin{align}
    \dot{\tilde{u}}_{ia} &= -k_{iu} x_{i1} +  \bar{d}_{i1}^\prime ,\label{eq47_} \\
    \dot{\tilde{q}}_i &= -k_{iq} x_{i2} - \tilde{\theta} _{ia} +  \bar d_{i4}, \label{eq48} \\
    \dot{\tilde{r}}_i &= -k_{ir} x_{i3}  - \tilde{\psi}_{ia}/\cos{\theta _i} + \bar d_{i5}. \label{eq49}
\end{align}
Here, $x_{i1}$, $x_{i2}$ and $x_{i3}$ are the extended states, due to the introduced shunting compensators, and are governed by the equations \eqref{eq44}, \eqref{eq45_} and \eqref{eq47}, respectively; $\bar{d}_{i1}^\prime = \bar{d}_{i1}-\dot{u}_{ia}^{cmd}$, $\bar{d} _{i4} = d_{i4}/m_{i4} - \dot{q_i}^{cmd}$ and $\bar{d} _{i5} = d_{i5}/m_{i5} - \dot{r_i}^{cmd}$ are the lumped disturbances.

We shall provide a theorem to establish the stability properties of the closed-loop system with all error states $(\tilde \theta _{ia}, \tilde \psi _{ia},  \tilde u_{ia}, \tilde q_i, \tilde r_i,  x_{i1}, x_{i2}, x_{i3})$ given the proposed controller \eqref{eq43_}--\eqref{eq47}. To facilitate the analysis, we shall follow a 2-step demonstration.

\textit{Step 1: Input-to-State Stability of $(\Tilde{\theta}_{ia},\tilde{\psi}_{ia})$-subsystem.}

The $(\Tilde{\theta}_{ia},\tilde{\psi}_{ia})$-dynamics are given in \eqref{eq30} and \eqref{eq31}. It is clear by control theory that as long as the control gains $k_{i \theta}$ and $k_{i \psi}$ are chosen as the positive numbers, the origin of the resulting system is input-to-state stable with respect to the inputs $\tilde{q}_i$ and $\tilde{r}_i / \cos{\theta _i}$. In particular, the following inequalities hold
\begin{align}
    \| \tilde{\theta}_{ia}\| &\le {e^{ - k_{i\theta} t}}\| {\tilde{\theta}_{ia}\left( 0 \right)} \| + \frac{1}{k_{i\theta}} \mathop {\sup }\limits_{0 \le \tau  \le t} \left\| {\tilde{q}_{i}\left( \tau  \right)} \right\|,  \label{eq_rv1} \\
    \| \tilde{\psi}_{ia} \| &\le {e^{ - k_{i\psi} t}} \| {\tilde{\psi}_{ia}\left( 0 \right)} \| + \frac{1}{k_{i\psi}} \mathop {\sup }\limits_{0 \le \tau  \le t} \lVert {  \frac{\tilde{r}_{i}\left( \tau  \right)}{\cos{\theta _i (\tau)}}} \rVert . \label{eq_rv2} 
\end{align}

\textit{Step 2: Input-to-State Stability of $(\tilde u_{ia}, x_{i1}, \tilde q_i, x_{i2}, \tilde r_i, x_{i3})$-subsystem.}

Let $\xi _{i1} =[\tilde u_{ia},  x_{i1}]^{\rm T}$, $\xi _{i2}= [ \tilde q_i, x_{i2}]^{\rm T}$, and $\xi _{i3} = [\tilde r_i, x_{i3}]^{\rm T} $. Define $ \bar{d}_{i2}^{\prime} = \bar{d}_{i4} -\tilde{\theta}_{ia} $ and $\bar{d}_{i3}^{\prime} =\bar{d}_{i5} -\tilde{\psi}_{ia}/\cos{\theta _i}$. The $\xi _{ij}$-dynamics $(j = 1,2,3)$ can be rewritten into the following  form 
\begin{align}
    \dot \xi _{ij} &= T_{j} \xi _{ij} + D_j \bar{d}_{ij}^{\prime}, \label{eq_rv3} 
\end{align}
with
\begin{align}
    T_{1}& = \begin{bmatrix}
    0 & -k_{iu}  \\
    g_{i1} & -\bar a_{i1}    
    \end{bmatrix}, \quad
    T_{2} = \begin{bmatrix}
    0 & -k_{iq} \\
    g_{i2} & -\bar a_{i2}
    \end{bmatrix},  \nonumber \\
     T_{3} & = \begin{bmatrix}
    0 &-k_{ir}\\
    g_{i3} & -\bar a_{i3} \\
    \end{bmatrix}, \quad
    D_1  = D_2 = D_3  =\begin{bmatrix}
    1  \\
    0 
    \end{bmatrix}. \nonumber
\end{align}
Notice that in matrix $T_{1}$ the constant $g_{i1}$ takes value of either $b_{i1}$ if $\tilde u_{ia} \ge 0$, or $b^\prime_{i1}$ otherwise, and accordingly, same as the matrices $T_{2}$ and $T_{3}$; $\bar a_{i1} = a_{i1}+\left| \tilde u_{ia} \right|$, $\bar a_{i2} = a_{i2}+\left| \tilde q_i \right|$, and $\bar a_{i3} = a_{i3}+\left| \tilde r_i \right|$.

\begin{lemma} \label{lm5}
Consider subsystems as in \eqref{eq_rv3}. If systems matrices $T_1$, $T_2$ and $T_3$ are made Hurwitz by suitably choosing the control design parameters $k_{iu}$, $k_{iq}$, $k_{ir}$, $a_{i1}$, $a_{i2}$, $a_{i3}$, $b_{i1}$, $b_{i2}$, $b_{i3}$, $b^\prime_{i1}$, $b^\prime_{i2}$ and $b^\prime_{i3}$. Then the above subsystems are all input-to-state stable in $t \in \left[0, \infty \right) $ with respect to the inputs $\bar{d}_{i1}^{\prime}$, $\bar{d}_{i4}^{\prime}$ and $\bar{d}_{i5}^{\prime}$, respectively.
\end{lemma}
\begin{proof}
Due to the fact that the subsystems \eqref{eq_rv3} are all linear, their solutions can be readily obtained as 
\begin{align}
    { \xi _{i j}}\left( t \right) = {e^{ T_{j}t}}{ \xi _{i j}}\left( 0 \right) + \int_0^t {{e^{ T_{j}\left( {t - \tau } \right)}} D_{j}\bar{d}_{ij}^{\prime} d\tau }  \quad (j = 1,2,3), \label{eq_rv6} 
\end{align}
where $j$ indicates different subsystems. Since $T_{j}$ is Hurwitz, we have the inequality $\|  e^{ T_{j}t}\| \le  c_{1j} e^{ -c_{2j}t}$, in which $c_{1j}$ and $c_{2j}$ are some positive constants and, in particular, $-c_{2j}$ is greater than the real part of the maximum eigenvalue of $T_j$. It is noted that while the matrices $T_i$ may be time-varying, their Hurwitz properties can still be maintained, which can be easily verified by their analytical solutions of the eigenvalues. By this, it then follows from \eqref{eq_rv6} that
\begin{align}
    {\| \xi _{i j}}\left( t \right) \|&\le c_{1j} {e^{ -c_{2j}t}}\|{\xi _{i j}}\left( 0 \right)\| + \int_0^t {{c_{1j} e^{ -c_{2j}\left( {t - \tau } \right)}} D_{j}\bar{d}_{ij}^{\prime} d\tau } \nonumber \\   
    &\le c_{1j} {e^{ -c_{2j}t}}\|{\xi _{i j}}\left( 0 \right)\| + \frac{1}{c_{2j}} \mathop {\sup }\limits_{0 \le \tau  \le t} \left\|  \bar{d}_{ij}^{\prime}  \right\|.  \label{eq_rv7} 
\end{align}
It can be concluded from \eqref{eq_rv7} that $\xi _{ij}$-subsystems are all input-to-state stable with respect to the inputs  $\bar{d}_{ij}^\prime$ in $t \in \left[0, \infty \right) $. This completes the proof.
\end{proof}

By far, we have shown in an independent way that both $(\Tilde{\theta}_{ia},\tilde{\psi}_{ia})$-subsystem and $(\tilde u_{ia}, x_{i1}, \tilde q_i, x_{i2}, \tilde r_i, x_{i3})$-subsystem are of input-to-state stability properties. However, it is observed that some of the above subsystems are coupled in their states. Thus, it is necessary to show the stability property of the overall coupled system.  We provide the following theorem to establish this.
\newtheorem{theorem}{Theorem} 
\begin{theorem} \label{thm1}
Consider together the subsystem \eqref{eq30}, \eqref{eq31}, \eqref{eq44} and \eqref{eq_rv3}. Suppose that there exist some positive constant $\alpha _2$ such that the lumped disturbances $\|\bar{d}_{i1}^{\prime}\|+ \|\bar{d}_{i4}\| +\|\bar{d}_{i5} \| \le  \alpha _2$, and also that the results obtained in the \textit{Step 1} and \textit{Step 2} hold. The states of all above subsystems are input-to-state stable if $1/c_{22} k_{i \theta} < 1$ and $ \alpha _{\theta} ^{{\prime 2}}/ c_{23}k_{i \psi} < 1$, where $\alpha _{\theta}^{\prime}= \sup _{t\ge 0}1/\cos{\theta _i}$.
\end{theorem}

\begin{proof} \label{pf3}
For the ease of illustration, introduce the $L_\infty$ norm for a signal $u(t)$ as below
\begin{align}
    \| u(t) \|_{L_\infty} = \sup_{t \ge 0} \| u(t)\|. \label{eq_64}
\end{align}
It is immediate from \eqref{eq_rv7} that
\begin{align}
    \| \tilde{q}_i(t) \| &\le \| \xi _{i 2}\left( t \right) \| \le \frac{1}{c_{22}} \mathop {\sup}\limits_{0 \le \tau  \le t} \|  \tilde{ \theta}_{ia}  \|  + \alpha_q + \frac{1}{c_{22}} \alpha _2, \nonumber \\
    \| \tilde{r}_i(t) \| &\le \| \xi _{i 3}\left( t \right) \| \le \frac{1}{c_{23}} \mathop {\sup }\limits_{0 \le \tau  \le t} \| \frac{\tilde{\psi}_{ia}}{\cos{\theta _i}}\|  + \alpha_r + \frac{1}{c_{23}} \alpha _2, \nonumber
\end{align}
with $\alpha_q = c_{12}\|{\xi _{i 2}}\left( 0 \right)\|$, $\alpha_r = c_{13} \|{\xi _{i 3}}\left( 0 \right)\|$, and applying $L_{\infty}$ norm, together with \eqref{eq_rv1} and \eqref{eq_rv2}, yield
\begin{align}
    \| \tilde{q}_i(t) \|_{L_\infty} &\le \frac{1}{c_{22}}  \|  \tilde{ \theta}_{ia}(t) \|_{L_{\infty}} + \alpha_q + \frac{1}{c_{22}} \alpha _2 \nonumber \\
    &\le \frac{1}{c_{22}}  ( \frac{1}{k_{i \theta}} \| \tilde{q}_i(t) \|_{L_{\infty}} + \alpha _\theta) + \alpha_q + \frac{1}{c_{22}} \alpha _2 \nonumber \\
    & = \frac{1}{c_{22} k_{i \theta}} \| \tilde{q}_i(t) \|_{L_{\infty}} + \alpha_q + \frac{1}{c_{22}}  ( \alpha _\theta +\alpha _2),  \label{eq63}
\end{align}
with $\alpha _\theta = \| {\tilde{\theta}_{ia}\left( 0 \right)} \|$, and
\begin{align}
    \| \tilde{\theta}_{ia}(t) \|_{L_\infty} &\le  \frac{1}{k_{i\theta}} \left\| {\tilde{q}_{i}\left( t \right)} \right\|_{L_\infty} + \alpha _\theta, \nonumber \\
    &\le \frac{1}{c_{22} k_{i\theta}} \|  \tilde{ \theta}_{ia}(t) \|_{L_{\infty}} + \alpha _{\theta} + \frac{1}{k_{i\theta}} (\alpha_q + \frac{1}{c_{22}} \alpha _2). \label{eq64}
\end{align}
Considering the condition $1/c_{22} k_{i \theta} < 1$, we further get
\begin{align}
    \| \tilde{q}_i(t) \|_{L_\infty} &\le  \mu _2 \left(\alpha_q + \frac{1}{c_{22}}  ( \alpha _\theta +\alpha _2) \right), \label{eq65} \\
    \| \tilde{\theta}_{ia}(t) \|_{L_\infty} &\le \mu _2 \left(    \alpha _{\theta} + \frac{1}{k_{i\theta}} (\alpha_q + \frac{1}{c_{22}} \alpha _2)  \right),\label{eq66} 
\end{align}
with 
\begin{align*}
    \mu _2 = \frac{1}{1- 1/ c_{22} k_{i \theta}}.
\end{align*}
Employing a similar argument, the bounds on $ \| \tilde{r}_i (t)\|$ and $\| \tilde{\psi}_{i a}(t)\|$ can be estimated as follows
\begin{align}
    \| \tilde{r}_i(t) \|_{L_\infty} &\le  \mu_3 \left(\alpha_r + \frac{1}{c_{23}}  ( \alpha _{\theta}^\prime \alpha _\psi +\alpha _2) \right), \label{eq67} \\
    \| \tilde{\psi}_{ia}(t) \|_{L_\infty} &\le \mu _3 \left(    \alpha _{\psi} + \frac{\alpha _\theta^\prime}{k_{i\psi}} (\alpha_r + \frac{1}{c_{23}} \alpha _2)  \right),\label{eq68} 
\end{align}
with $\alpha _\psi = \| {\tilde{\psi}_{ia}\left( 0 \right)} \|$ and
\begin{align*}
    \mu _3 = \frac{1}{1-  \alpha _{\theta} ^{{\prime 2}}/ c_{23} k_{i \psi}}.
\end{align*}
As per the results obtained in \textit{Steps 1} and \textit{2} , together with the derived boundedness properties of signals $\tilde{\theta}_{ia}$, $\tilde{\psi}_{ia}$, $\tilde{q}_i$ and $\tilde{r}_i$, it is readily concluded that all the states of the subsystems \eqref{eq30}, \eqref{eq31}, \eqref{eq44} and \eqref{eq_rv3} are input-to-state stable. This completes the proof.
\end{proof}

\begin{remark} \label{rmk10}
It is worthwhile noting that in Theorem \ref{thm1} the boundedness conditions on $\dot{u}_{ia}^{cmd}$, $\dot q_i^{cmd}$ and $\dot r_i^{cmd}$ are used, which is in effect easy to verify, due to the fact that the functions $u_{ia}^{cmd}$, $r_i^{cmd}$ and $q_i^{cmd}$ defined in \eqref{eq12__}, \eqref{eq28} and \eqref{eq29}, respectively, are locally Lipschitz over a domain of interest, and such a local property can be maintained by the input-to-state stability results.
\end{remark}

\begin{remark} \label{rmk9}
It can be seen from the estimated bounds on error states that the robustness against the grouped disturbances depends mainly on the parameters $k_{i \theta}$, $k_{i \psi}$ and $c_{2j} \ (j =1,2,3)$; in particular, as mentioned $-c_{2j}$ are greater than the real part of the maximum eigenvalues of the $T_j$, which depends not only on $k_u$,  $k_q$, $k_r$, but also on $a_{ij}$, $b_{ij}$ and $b_{ij}^{\prime}$. Thus, it is possible to avoid taking large values for $k_u$,  $k_q$, $k_r$, while achieving good robustness. 
\end{remark}

\begin{remark} \label{rm12}
Note also that the shunting compensator introduced is essentially a dynamic model acting as a low-pass filter. Thus, except for the disturbance attenuation, it behaves well in terms of noise rejection and control smoothing. Besides, it is found that the outputs of the shunting compensators can be bounded upper by $b_{ij}$ and lower by $b_{ij}^\prime$, and hence the actuator saturation issue can be resolved. All of these properties show that the developed controller outperforms the conventional backstepping based methods.
\end{remark}

The overall stability of proposed UUVs formation system can be readily established using the following corollary. 
\newtheorem{corollary}{Corollary}
\begin{corollary} \label{corol1}
Under the assumptions of Lemma \ref{lm2} and Theorem \ref{thm1}, the proposed formation tracking system of underactuated underwater vehicles fleet is input-to-state stable with respect to the disturbances $d_i$, $i = 1,2,\ldots, N$ as input. 
\end{corollary}
\begin{proof}
 In Theorem \ref{thm1}, it is demonstrated that given the neuro-dynamics based backstepping controllers, as developed in \eqref{eq43_}--\eqref{eq47}, the error variables, i.e., $\tilde \theta _{ia}$, $\tilde \psi _{ia}$, $\tilde u_{ia}$, $\tilde q_i$ and $\tilde r_i$  $(i = 1,2,\ldots, N)$, can be made uniformly bounded with respect to the disturbances so long as the control parameters associated are chosen properly. Then, it is immediate by invoking Lemma \ref{lm2} that the overall consensus formation tracking error $e$ of the UUVs fleet is uniformly ultimately bounded. Also, due to the equation \eqref{eq15__} we conclude that $\dot e$ is bounded as well, and thus together with the Assumption \ref{asp2} as well as the kinematic equation \eqref{eq1} of UUVs, it is clear that all the states of the UUVs are uniformly ultimately bounded with respect to the disturbances. This completes the proof.
\end{proof}

\section{SIMULATION RESULTS}\label{s5}
This section presents several numerical simulations to illustrate the effectiveness of the proposed constrained formation protocol as well as the neurodynamics-based robust backstepping controller. In the simulations, four underactuated autonomous underwater vehicles are employed to construct a formation system. The aim of the system is that by leveraging the equipped onboard formation controller each UUV in the group can be steered to track a desired common 3D straight line and, meanwhile, a predefined quadrilateral formation pattern can also be formed and maintained.

Since the developed formation protocol is implemented in a distributed fashion, which means only locally neighboring information can be accessed by each UUV, the communication topology associated is depicted in Fig. \ref{fig_2}. The weights on the topological graph are selected as ${a_{21}} = {a_{32}}  ={a_{43}} = 1$, ${a_{12}} = a_{23} = a_{34} = 0.8$, and  ${b_1} = {b_2} = {b_3} = {b_4} = 1$. The equations of motion of UUVs are described by the equations \eqref{eq1} and \eqref{eq2} with following parameters (every parameter follows an international standard unit): $m_i = 10$, $I_{y,i}=3$, $I_{z,i}=2$, $\beta _{\dot{u},i} = 6$, $\beta _{\dot{v},i} = 1.1$, $\beta _{\dot{w},i} = 1.15$, $\beta _{\dot q,i} = 0.5$, $\beta _{\dot r,i} = 0.45$, $\beta _{u,i} = 1$, $\beta _{ v,i} = 1.1$, $\beta _{w,i} = 1.15$, $\beta _{q,i} = 0.2$, $\beta _{r,i} = 0.25$, and $\beta _{b,i} = 0.1$, $\left( {i \in \left\{ {1,2,3,4} \right\}}  \right)$. The desired 3D path is defined as $\eta _1^d\left( t \right) = {\left[0.7t +5, 0.1t+1, {5} \right]^{\rm T}}$, its derivative as $\dot \eta _1^d (t) = \left[ 0.7, 0.1,0 \right] ^{\rm T}$. To form a prescribed formation shape, the relative positions between UUVs are given by ${\Delta _{12}} = {\left[ {0,10,0} \right]^{\rm T}}$, ${\Delta _{21}} = {\left[ {0, - 10,0} \right]^{\rm T}}$, ${\Delta _{23}} = {\left[ { - 10,0,0} \right]^{\rm T}}$, ${\Delta _{32}} = {\left[ {10,0,0} \right]^{\rm T}}$, ${\Delta _{34}} = {\left[ {0, - 10,0} \right]^{\rm T}}$ and ${\Delta _{43}} = {\left[ {0,10,0} \right]^{\rm T}}$. The initial conditions of four UUVs are given as ${\eta _{1}(0)} = {\left[ {0,0,0,0,0} \right]^{\rm T}}$, ${\eta _2(0)} = {\left[ {-1,-10,0,0,0} \right]^{\rm T}}$, ${\eta _3(0)} = {\left[ {8.5,-10.1,0,0,0} \right]^{\rm T}}$, ${\eta _4(0)} = {\left[ {8.4,-0.1,0,0,0} \right]^{\rm T}}$, and ${\nu_i(0)} = \left[0.1,0,0,0,0\right]^{\rm T} ,\left( {i \in \left\{ {1,2,3,4} \right\}} \right)$. In order to make the simulation result more convincing, five controllers' performances are compared, that is, neurodynamics-based backstepping optimal controller (NBOC, i.e., the proposed approach),  backstepping optimal controller (BOC, i.e., without neurodynamics), neurodynamics-based backstepping controller (NBC, i.e., without online optimization), backstepping controller (BC), and backstepping sliding mode controller (BSMC). The control parameters of five controllers are listed in the TABLE \ref{t1}.

\begin{table*}[!htbp] 
\centering
\caption{Control parameters}
\begin{tabular}{cccccc}  
\toprule  
Parameters & NBOC & BOC & NBC & BC & BSMC \\
\midrule  
$K_{1,i}$ & diag(0.3,0.3,0.3) & diag(0.3,0.3,0.3) & diag(0.6,0.6,0.6) & diag(0.6,0.6,0.6) & diag(0.6,0.6,0.6)\\
$K_{2,i}$ & diag(10,10,10)&  diag(10,10,10)& diag(10,10,10) & diag(10,10,10)  & diag(20,15,15)\\
$Q$ & diag(10,10,10)& diag(10,10,10) & N/A& N/A & N/A\\
$R_1$ & diag(1,1,1)& diag(1,1,1) & N/A& N/A & N/A\\
$R_2$ & diag(1,1,1)& diag(1,1,1) & N/A& N/A & N/A\\
$p_1$,$p_2$,$p_3$ & 0.1,0.1,0.1& 0.1,0.1,0.1 & N/A& N/A & N/A\\
$a_i$ & 10  & N/A & 10 & N/A & N/A\\
$b_i$ & 30 & N/A & 30 & N/A & N/A\\
$d_i$ & 30 & N/A & 30 & N/A & N/A\\
\bottomrule 
\end{tabular}
\label{t1}
\end{table*}

\begin{figure}[!t]
\centering
\includegraphics[width=3in]{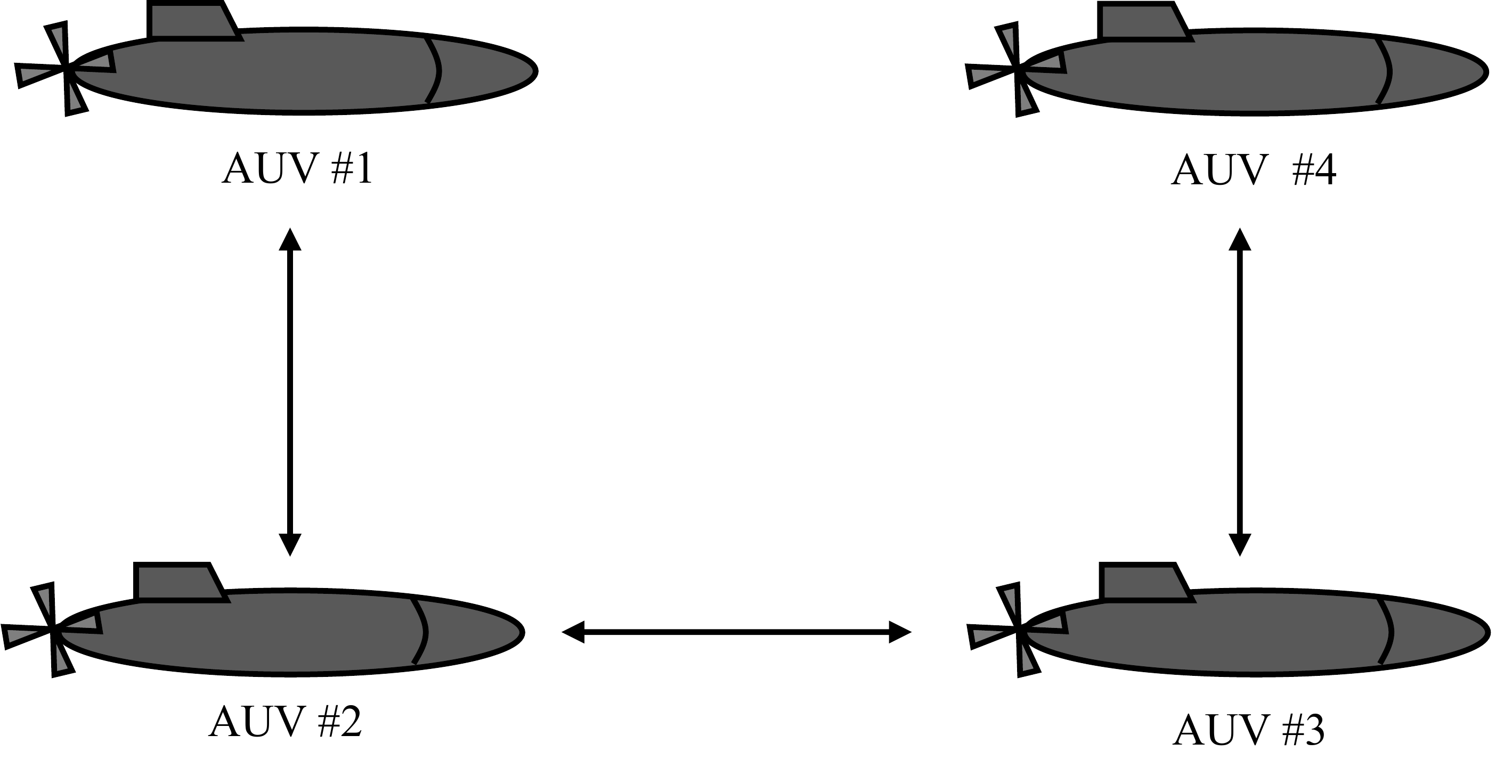}
\caption{The communication topology graph for the consensus formation tracking of 4 UUVs.}
\label{fig_2}
\end{figure}

\begin{figure}[!htbp]
\centering
\includegraphics[width=1\columnwidth]{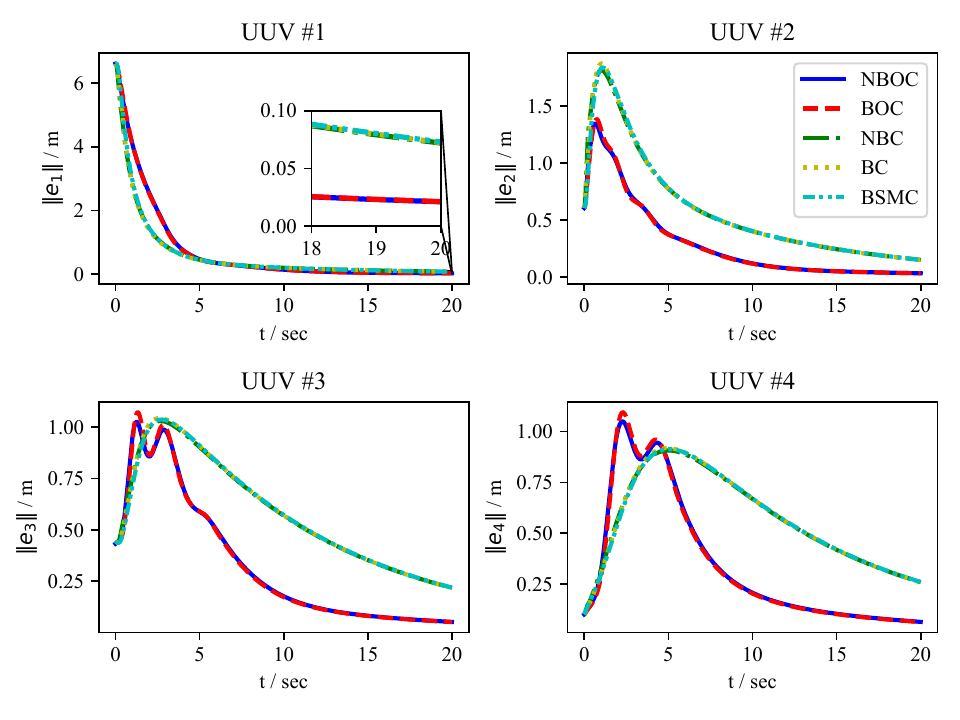}
\caption{The consensus formation tracking errors of 4 UUVs.}
\label{fig_3}
\end{figure}

\begin{figure}[!htbp]
\centering
\includegraphics[width=1\columnwidth]{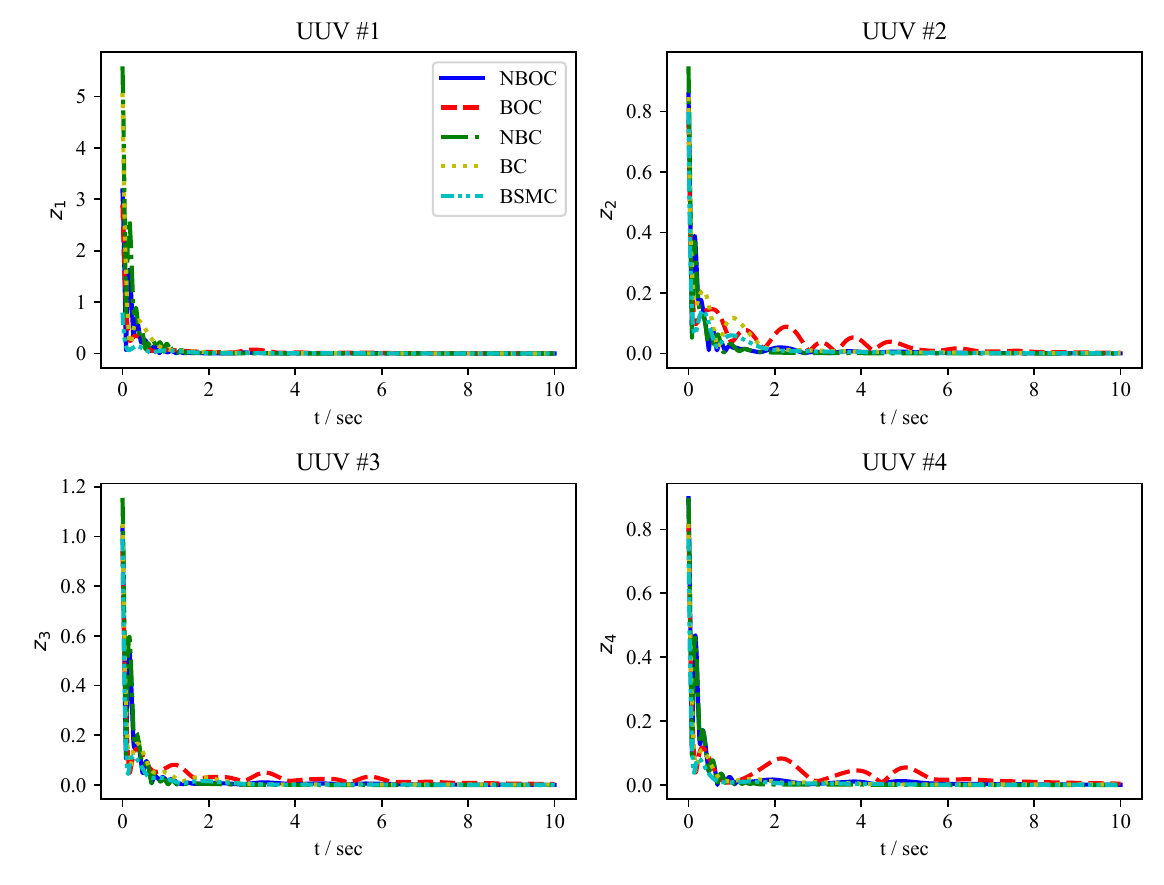}
\caption{The velocity tracking errors of 4 UUVs ($z_i = \sqrt{\tilde{u}_{ia}^2 + \tilde{q}_i^2 + \tilde{r}_i^2}$).}
\label{fig_4}
\end{figure}

\begin{figure}[!htbp]
\centering
\includegraphics[width=1\columnwidth]{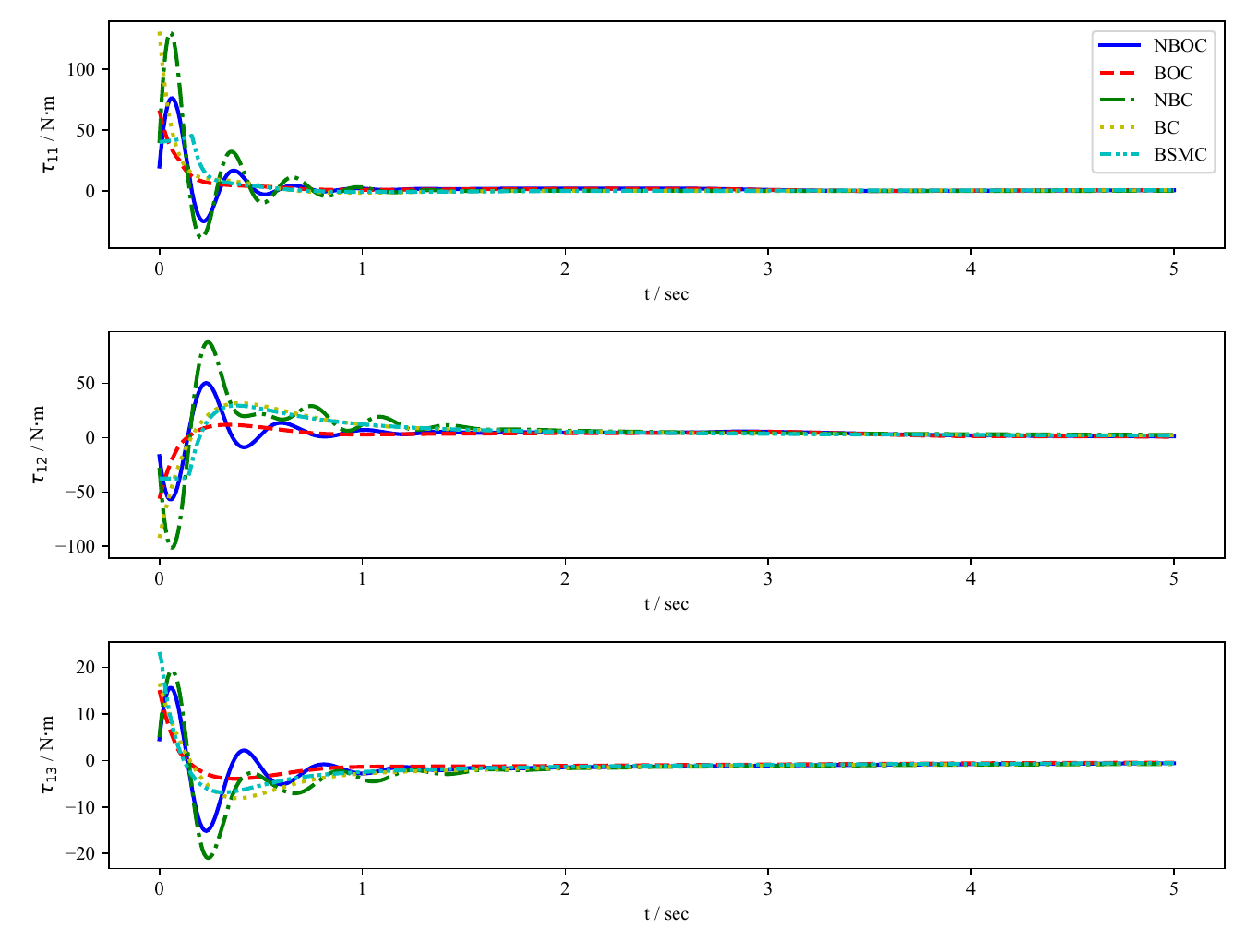}
\caption{The control inputs of UUV $1$.}
\label{fig_5}
\end{figure}

\begin{figure}[!htbp]
\centering
\includegraphics[width=1\columnwidth]{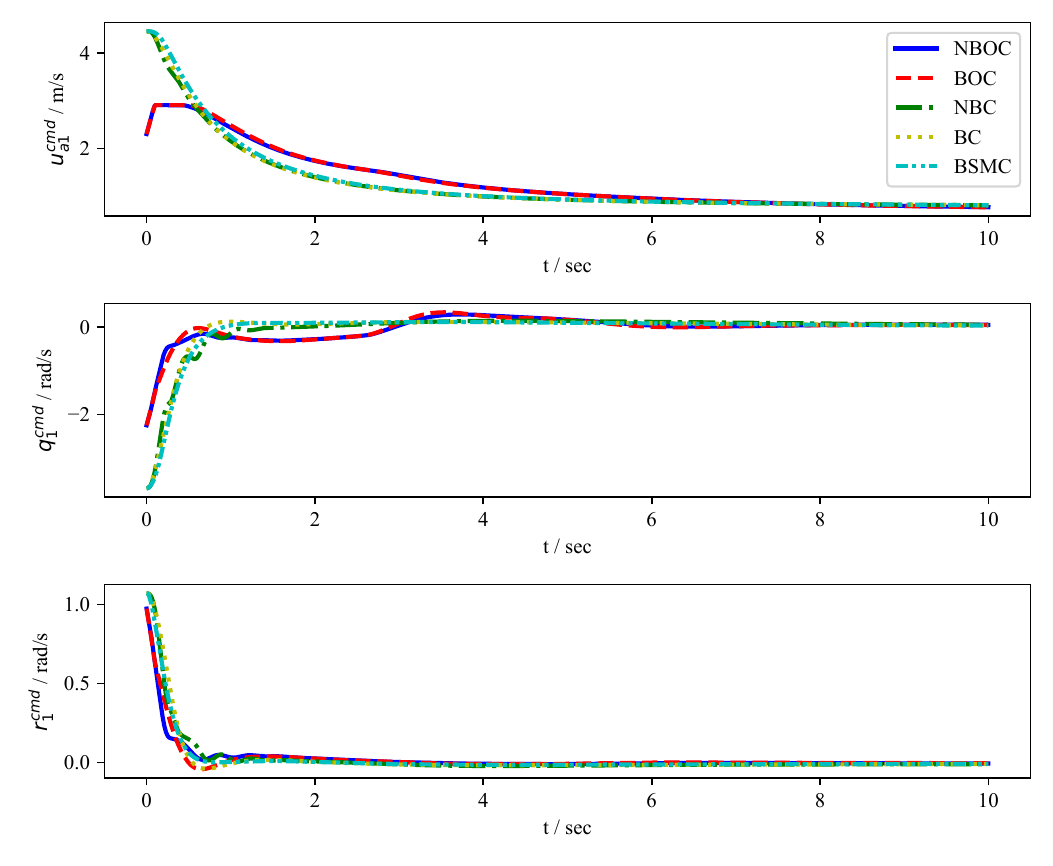}
\caption{The velocity commands of UUV $1$.}
\label{fig_6}
\end{figure}

\begin{figure}[!htbp]
\centering
\includegraphics[width=1\columnwidth]{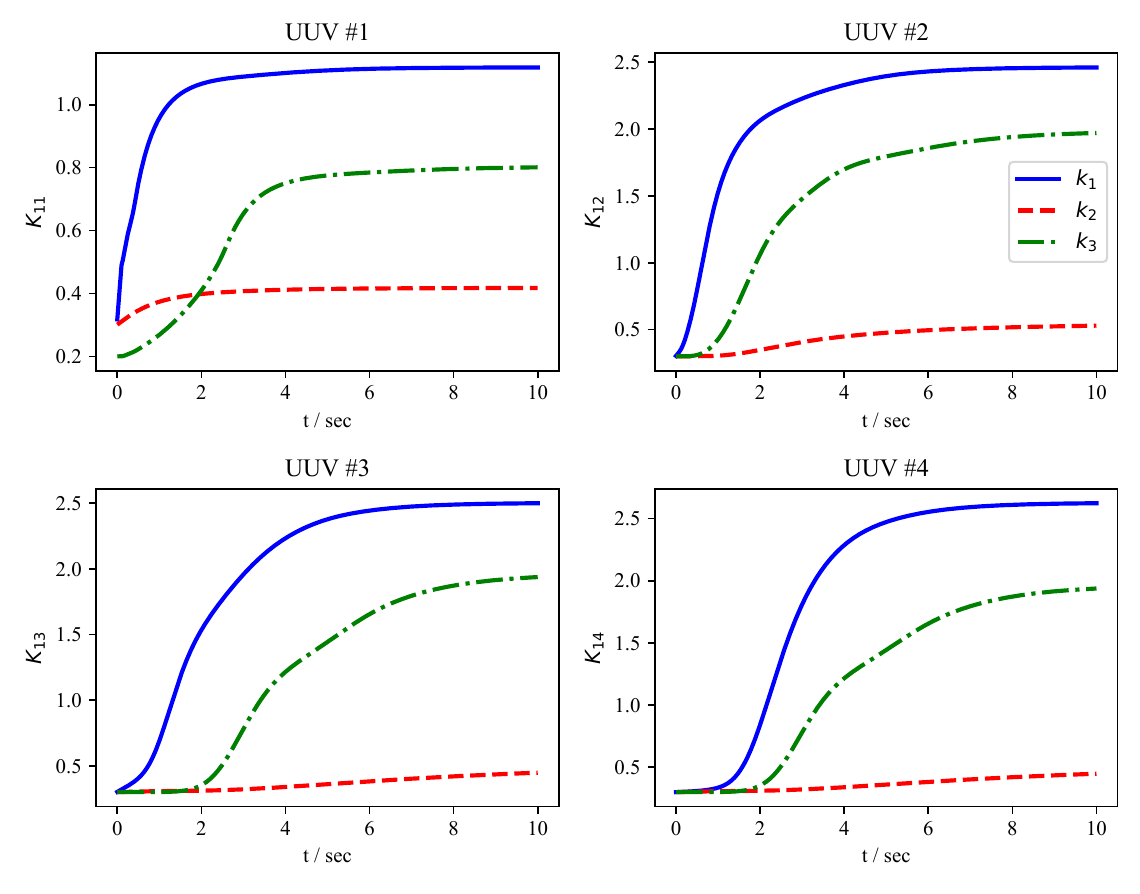}
\caption{The evolution of virtual control gains of 4 UUVs in NBOC.}
\label{fig_7}
\end{figure}

In the first case, there are no external disturbances added to the vehicles. It can be seen from the Figs. \ref{fig_3} and \ref{fig_4} that the formation tracking objective is perfectly achieved by all of the formation controllers. In particular, the controllers fitted with online optimization exhibit a faster rate of convergence as seen clearly from the behaviors of UUVs 2-4, while the control efforts needed are as nearly twice small as the NBC and BC approaches at the starting time, as shown in Fig. \ref{fig_5} (note that for conciseness only the UUV 1's control activities are presented, and actually the rest of vehicles behave much similar). In addition to that, another significant advantages of the online optimization are that it avoids an evident speed jump and, meanwhile, the velocity commands generated are confined within a given interval as observed in Fig \ref{fig_6}; in contrast, NBC, BC and BSMC methods all yield a relatively large velocity necessity in the beginning, due to the initial consensus errors. The properties obtained by the motion optimization are important for the controller design, since all of the real UUVs have their physical limitations on maneuvering capability. The optimization processes of the NBOC method for each UUV are presented in Fig. \ref{fig_7}, from which an automatic adjustment for the virtual control gains can be observed.

\begin{figure}[!htbp]
\centering
\includegraphics[width=1\columnwidth]{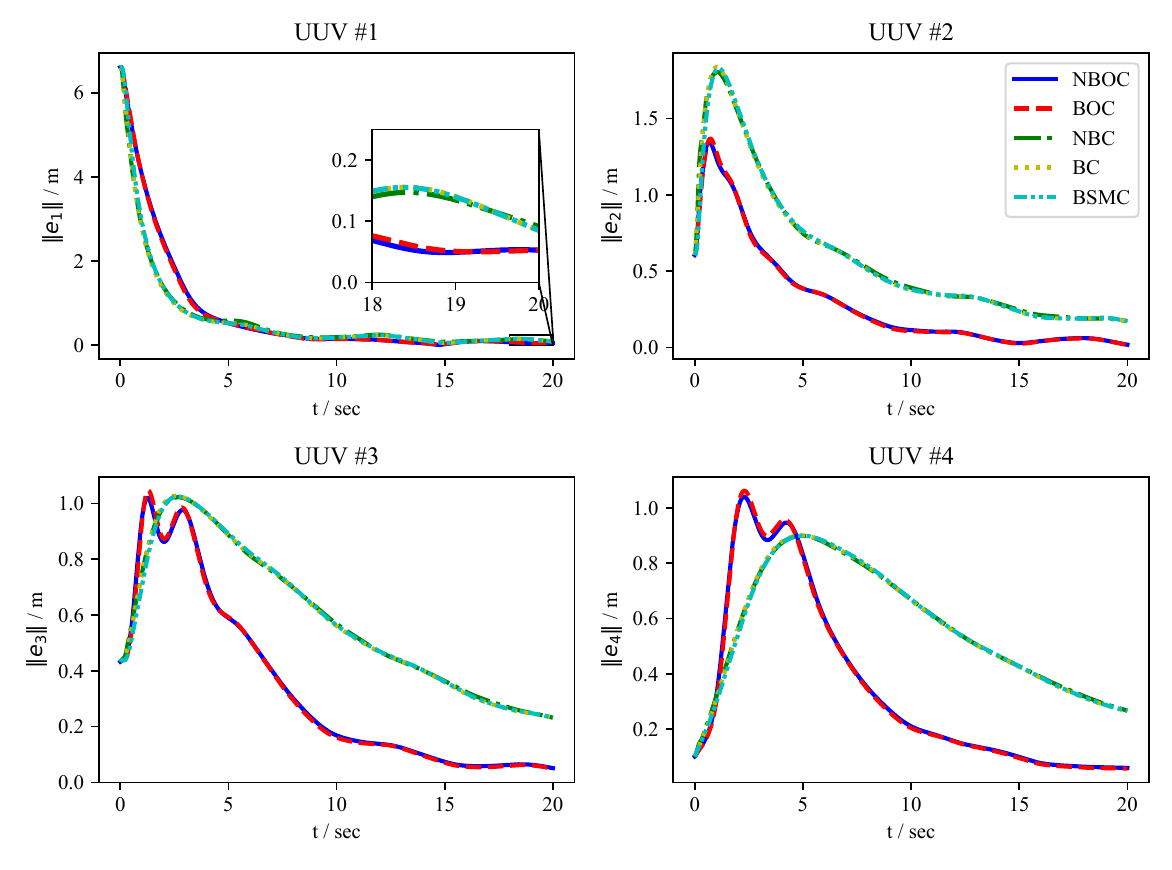}
\caption{The consensus formation tracking errors of 4 UUVs applied with disturbances.}
\label{fig_8}
\end{figure}

\begin{figure}[!htbp]
\centering
\includegraphics[width=1\columnwidth]{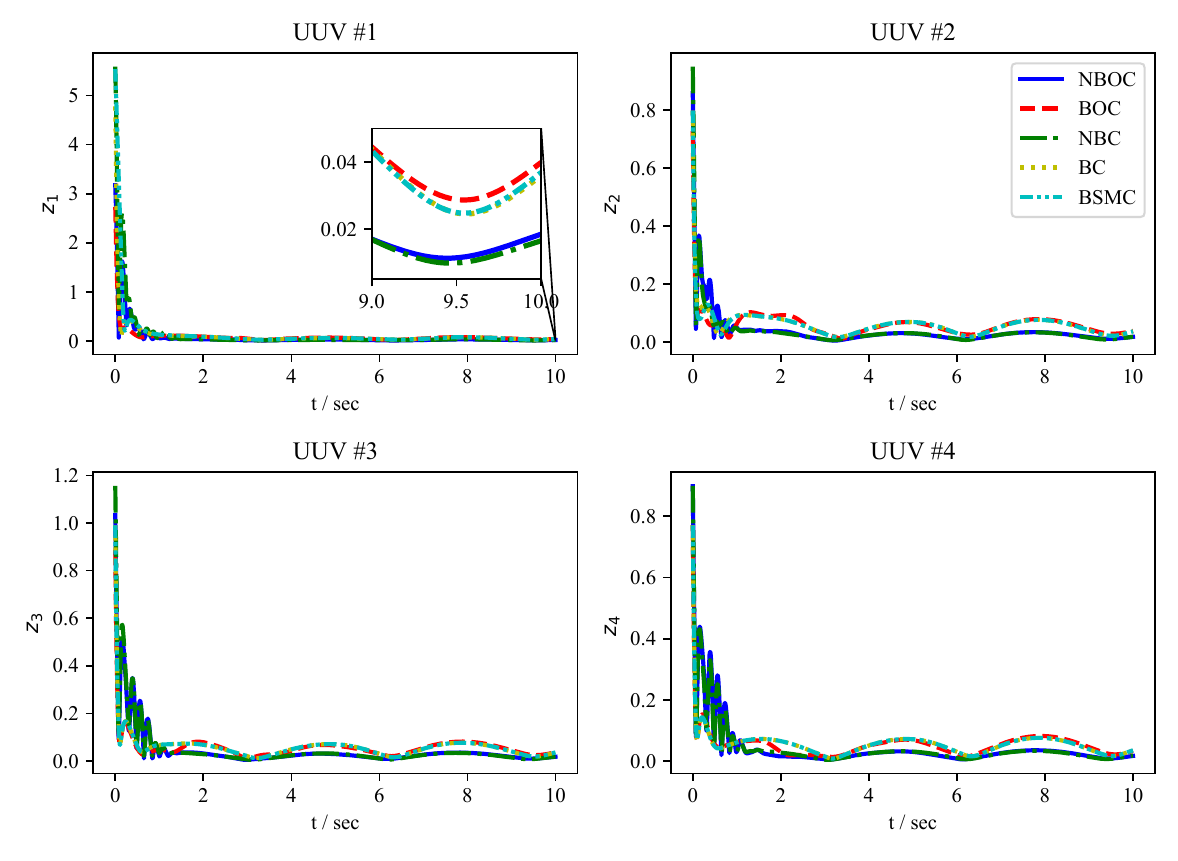}
\caption{The velocity tracking errors of 4 UUVs applied with disturbances.}
\label{fig_9}
\end{figure}

\begin{figure}[!htbp]
\centering
\includegraphics[width=1\columnwidth]{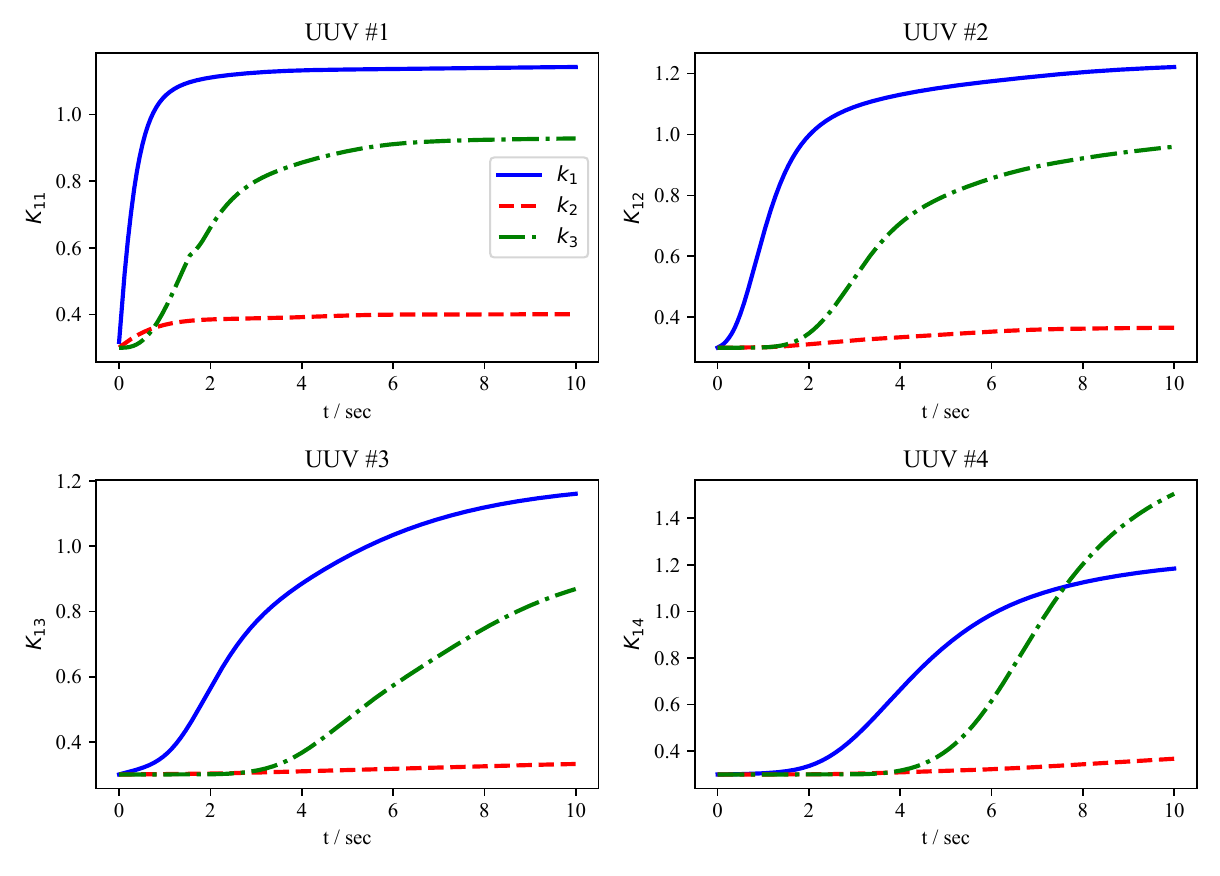}
\caption{The evolution of virtual control gains of 4 UUVs in NBOC applied with disturbances.}
\label{fig_10}
\end{figure}

In order to verify the robustness of the proposed formation control protocol, in the next case we inject the period exogenous disturbances into the four UUVs to simulate the influence of the ocean waves and currents. The disturbances applied are described by $d_i = \left[ {{3.1\sin \left( t \right),3.1\cos \left( t \right),2.1\sin \left( t \right)}}  \right.$, $\left. {{1.1\sin \left( t \right), 1.1\sin \left( t \right)}}  \right]$, $\left( {i \in \left\{ {1,2,3,4} \right\}} \right)$. The formation performances of five control methods under disturbances are plotted in Figs. \ref{fig_8} and \ref{fig_9}. Similar to the unperturbed situation, NBOC and BOC methods (i.e., assisted with online optimization) show a faster convergence property as seen in Fig. \ref{fig_8} and, meanwhile, have smaller consensus tracking errors compared to the other approaches. It implies that the optimal virtual control commands developed exhibit a better robustness property when faced with unknown disturbances. At the dynamic level, as illustrated by Fig. \ref{fig_9}, the controllers equipped with the neurodynamics model render apparently smaller velocity tracking errors, thus suggesting that such methods possess good robustness in disturbance attenuation. The optimization processes of the NBOC method for each UUV are depicted in Fig. \ref{fig_10}, all of which show a smooth convergence behavior even in the presence of disturbances. Based on the above observations, the proposed neurodynamics-based backstepping controller nested with an online optimization procedure achieves the best formation performances over the other four methods in terms of convergence speed, steady state accuracy, disturbance attenuation, and constraint fulfillment.

\section{CONCLUSION}\label{s6}
This paper addresses the robust constrained consensus formation tracking problem for a fleet of underactuated autonomous underwater vehicles in 3D space. A spherical coordinate transformation is introduced, based on which a novel distributed optimal formation control protocol is synthesized by iteratively solving a designed constrained optimization problem. As such, an optimal performance index can be achieved while the constraints on UUVs velocities can be fulfilled. Then, the feasibility and stability of the optimization problem are discussed. In order to realize the optimal control commands efficiently, a neuro-dynamics based robust backstepping controller is designed. The issue of "explosion of terms" incurred in conventional backstepping controllers is addressed, and the control performance as well as robustness properties against unknown disturbances are improved. Furthermore, a rigorous stability proof of the proposed formation control method is performed to guarantee the desired performance at the theoretical level. Finally, extensive numerical simulations are carried out to further demonstrate the effectiveness and superiority of the developed UUVs formation protocol.



 
\bibliographystyle{IEEEtran}
\bibliography{reference}

\begin{IEEEbiography}[{\includegraphics[width=1in,height=1.25in,clip,keepaspectratio]{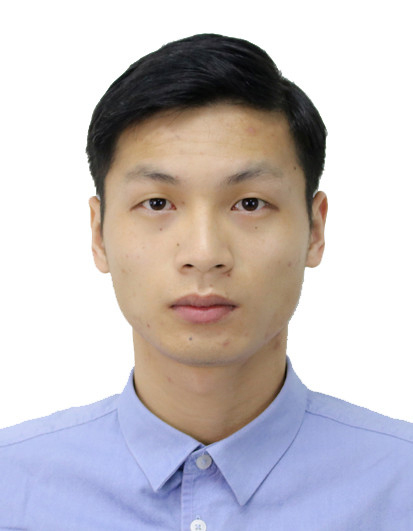}}]{Tao Yan}
(Graduate Student Member, IEEE) received the B.S. degree in automation from the North China Institute of Aerospace Engineering, Langfang, China, in 2016, and the M.S. degree in control science and engineering from the Zhejiang University of Technology, Hangzhou, China, in 2020. He is currently pursuing his Ph.D. degree at the University of Guelph, ON, Canada. His research interests include the nonlinear control, machine learning, distributed control and optimization, optimal estimation, and networked underwater vehicle systems.
\end{IEEEbiography}

\vskip  -2\baselineskip plus -1fil

\begin{IEEEbiography}[{\includegraphics[width=1in,height=1.25in,clip,keepaspectratio]{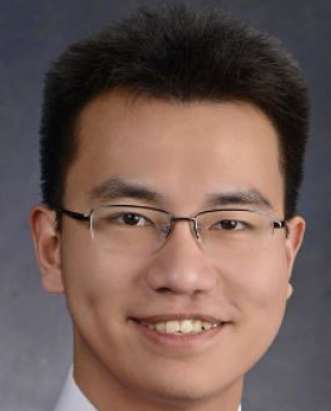}}]{Zhe Xu}
(Member, IEEE) received B.ENG. degree in Mechanical Engineering in 2018, and M.A.Sc. and Ph.D degree in Engineering Systems and Computing in 2019 and 2023, respectively, from University of Guelph. He is currently a post-doctoral fellow with Department of Mechanical Engineering at McMaster University. His research interests include networked systems, tracking control, estimation theory, robotics, and intelligent systems.
\end{IEEEbiography}

\vskip  -2\baselineskip plus -1fil

\begin{IEEEbiography}[{\includegraphics[width=1in,height=1.25in,clip,keepaspectratio]{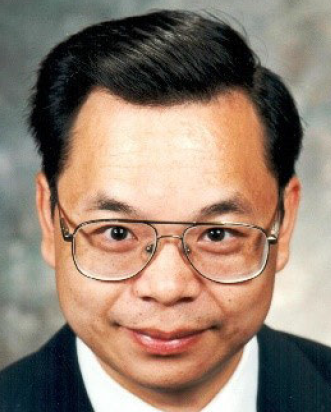}}]{Simon X. Yang}
(Senior Member, IEEE) received the B.Sc. degree in engineering physics from Beijing University, Beijing, China, in 1987, the first of two M.Sc. degrees in biophysics from the Chinese Academy of Sciences, Beijing, in 1990, the second M.Sc. degree in electrical engineering from the University of Houston, Houston, TX, in 1996, and the Ph.D. degree in electrical and computer engineering from the University of Alberta, Edmonton, AB, Canada, in 1999.  He is currently a Professor and the Head of the Advanced Robotics and Intelligent Systems (ARIS) Laboratory at the University of Guelph, Guelph, ON, Canada. His research interests include robotics, intelligent systems, control systems, sensors and multi-sensor fusion, wireless sensor networks, intelligent communication, intelligent transportation, machine learning, fuzzy systems, and computational neuroscience. 

Prof. Yang he has been very active in professional activities. He serves as the Editor-in-Chief of \textit{Intelligence \& Robotics}, and \textit{International Journal of Robotics and Automation}, and an Associate Editor of \textit{IEEE Transactions on Cybernetics}, \textit{IEEE Transactions on Artificial Intelligence}, and several other journals. He has involved in the organization of many international conferences.
\end{IEEEbiography}

\vspace{11pt}
\newpage
\begin{IEEEbiography}[{\includegraphics[width=1in,height=1.25in,clip,keepaspectratio]{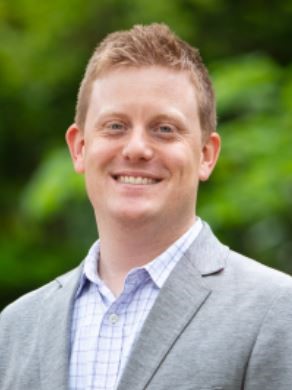}}]{S. Andrew Gadsden}
(Senior Member, IEEE) received the Ph.D. degree in mechanical engineering  from McMaster University, Hamilton, ON, Canada, in 2011. He is an Associate Professor with the Department  of Mechanical Engineering, McMaster University. He was an Associate/Assistant Professor with the University of Guelph, Guelph, ON, Canada, and the University of Maryland, College Park, MA, USA. His research area includes control and estimation theory, artificial intelligence and machine learning, and cognitive systems.   

Dr. Gadsden has been the recipient of numerous international awards and recognitions. In January 2022, he and his fellow air-LUSI project teammates were awarded the NASA’s Prestigious 2021 Robert H. Goddard Award in Science for their work on developing an airborne lunar spectral irradiance.  He is a certified Project Management Professional. He is an Associate Editor of \textit{Expert Systems with Applications} and is a reviewer for a number of ASME and IEEE journals and international conferences. He is an Elected Fellow of ASME.
\end{IEEEbiography}

\vfill

\end{document}